\documentclass[a4paper,twocolumn,10pt,accepted=2023-10-26]{quantumarticle}
\pdfoutput=1

\usepackage[dvipsnames]{xcolor}
\usepackage[colorlinks=true,bookmarks=false]{hyperref}
\usepackage{amsmath,amsfonts,amssymb,amsthm}
\usepackage{mathtools}
\usepackage{tabularx}
\usepackage{graphicx}
\usepackage[italicdiff]{physics}
\usepackage{enumitem}
\usepackage{setspace}
\usepackage[numbers]{natbib}
\numberwithin{equation}{section}

\newtheorem{theorem}{Theorem}
\newtheorem{corollary}{Corollary}
\newtheorem{proposition}{Proposition}
\newtheorem{lemma}{Lemma}

\newtheorem{definition}{Definition}

\newcommand{\avg}[1]{\langle#1\rangle}
\newcommand{\Avg}[1]{\left\langle#1\right\rangle}
\newcommand{\bs}[1]{\boldsymbol{#1}}

\newcommand{\bk}[1]{\left(#1\right)}
\newcommand{\Bk}[1]{\left[#1\right]}
\newcommand{\BK}[1]{\left\{#1\right\}}

\newcommand{\mc}[1]{\mathcal #1}
\newcommand{\gce}[1]{#1_*}
\newcommand{\HS}{\textrm{HS}}

\DeclareMathOperator*{\argmin}{arg\,min}

\DeclareMathOperator{\expect}{\mathbb E}

\DeclareMathOperator{\MSE}{MSE}
\newcommand{\guta}{Gu\c{t}\u{a}}
\begin{document}

\title{Operational meanings of a generalized conditional
  expectation in quantum metrology}

\author{Mankei Tsang}
\email{mankei@nus.edu.sg}
\orcid{0000-0001-7173-1239}

\homepage{https://blog.nus.edu.sg/mankei/}
\affiliation{Department of Electrical and Computer Engineering,
  National University of Singapore, 4 Engineering Drive 3, Singapore
  117583}

\affiliation{Department of Physics, National University of Singapore,
  2 Science Drive 3, Singapore 117551}

%\date{\today}
%\pacs{42.50.Wk, 03.65.Ta, 42.65.Yj}

\maketitle

\begin{abstract}
  A unifying formalism of generalized conditional expectations (GCEs)
  for quantum mechanics has recently emerged, but its physical
  implications regarding the retrodiction of a quantum observable
  remain controversial. To address the controversy, here I offer
  operational meanings for a version of the GCEs in the context of
  quantum parameter estimation.  When a quantum sensor is corrupted by
  decoherence, the GCE is found to relate the operator-valued optimal
  estimators before and after the decoherence. Furthermore, the error
  increase, or regret, caused by the decoherence is shown to be equal
  to a divergence between the two estimators. The real weak value as a
  special case of the GCE plays the same role in suboptimal
  estimation---its divergence from the optimal estimator is precisely
  the regret for not using the optimal measurement.  For an
  application of the GCE, I show that it enables the use of dynamic
  programming for designing a controller that minimizes the estimation
  error. For the frequentist setting, I show that the GCE leads to a
  quantum Rao-Blackwell theorem, which offers significant implications
  for quantum metrology and thermal-light sensing in particular. These
  results give the GCE and the associated divergence a natural,
  useful, and incontrovertible role in quantum decision and control
  theory.
\end{abstract}

\maketitle

\section{Introduction}
The conditional expectation is an essential concept in classical
probability and statistics \cite{parth05}. Given some observed data in
an experiment, the conditional expectation of a hidden random variable
is the best approximation of the hidden variable in a least-squares
sense and thus plays a central role in Bayesian estimation theory
\cite{parth05,berger}. Another important application is in the
Rao-Blackwell theorem \cite{rao45,blackwell47}, which exploits the
variance reduction property of the conditional expectation to improve
an estimator and has found widespread uses in classical statistics
\cite{lehmann98,thompson12,robert}.

Many attempts have been made over the past few decades to generalize
the concept of conditional expectation for quantum mechanics
\cite{umegaki54,personick71,belavkin73,accardi82,petz,aav,dressel15,hayashi,ohki15,ohki18,tsang22b}.
Umegaki's version for von Neumann algebra may be the earliest
\cite{umegaki54}. His axiomatic definition is so restrictive, however,
that his conditional expectation does not exist in many situations
\cite{holevo,petz}; this existence problem has led Holevo to remark
that ``conditional expectations play a less important part in quantum
than in classical probability'' \cite{holevo}.  In quantum estimation
theory, Personick \cite{personick71} and Belavkin and Grishanin
\cite{belavkin73} proposed an operator-valued estimator that is
optimal for Bayesian parameter estimation and can also be regarded as
a quantum conditional expectation. On the other hand, Accardi and
Cecchini proposed yet another conditional expectation for von Neumann
algebra \cite{accardi82}, which became instrumental in Petz's work on
quantum sufficient channels \cite{petz}. Many other investigations of
quantum conditional expectations can be found in the literature on
weak values \cite{aav,dressel15}, quantum filtering
\cite{belavkin_qnd,bouten,wiseman_milburn}, quantum retrodiction
\cite{barnett00,barnett21}, and quantum smoothing
\cite{yanagisawa,smooth,smooth_pra1,smooth_pra2,guevara,chantasri21,ohki15,ohki18,tsang22b}. In
recent years, it has been recognized \cite{ohki15,ohki18,tsang22b}
that many of these quantum conditional expectations can be unified
under a mathematical formalism of generalized conditional expectations
(GCEs) \cite{hayashi}. The GCE formalism can also be rigorously
connected to the concepts of quantum states over time and generalized
Bayes rules \cite{leifer13,horsman17}, as shown by Parzygnat and
Fullwood \cite{parzygnat23}.

Despite the mathematical progress, the GCEs have provoked fierce
debates regarding their physical meaning and usefulness, especially
when it comes to the weak values
\cite{ferrie_combes,ferrie_combes_coin,vaidman14,jordan14,gough19,gough20,james21}.
The debates center on two issues: whether it makes any sense to
estimate the value of a quantum observable in the past (retrodiction)
and whether the GCEs offer any use in quantum metrology, where quantum
sensors are used to estimate classical parameters.  This work
addresses both issues by demonstrating how a certain version of the
GCEs---of which the real weak value is a special case---can play
fundamental roles in quantum parameter estimation in both Bayesian and
frequentist settings.

When a quantum sensor suffers from decoherence, I show that the GCE
relates the two Personick estimators before and after the decoherence.
Moreover, the error increase due to the decoherence, henceforth called
the regret, is shown to be equal to a divergence measure between the
two estimators.  By regarding a suboptimal measurement as a
decoherence process, I show that the weak value is a special case of
the GCE and its divergence from the Personick estimator is precisely
the regret due to the measurement suboptimality. For the frequentist
setting, I also propose a quantum Rao-Blackwell theorem based on the
GCE.

These fundamental results lead to many significant consequences in
quantum metrology. To wit, the Markovian nature of the GCE is shown to
enable the use of dynamic programming \cite{bertsekas} for optimizing
a measurement protocol, while
Corollaries~\ref{cor_mono}--\ref{cor_therm} in this work reveal the
monotonicity of the Bayesian error, the optimality of von Neumann
measurements in Bayesian and frequentist settings, the optimality of
symmetric estimators for symmetric states, the optimality of
direct-sum estimators for direct-sum states, and the optimality of
photon counting for certain thermal states. A key feature of these
optimality results is that they are direct statements about the
mean-square errors and are valid for both biased and unbiased
estimators, unlike many results based on Cram\'er-Rao-type bounds,
which require heavy assumptions about the estimators and the density
operators.

This paper is organized as follows. To set the stage and make the
paper self-contained, Sec.~\ref{sec_review} reviews the concept of
GCEs, emphasizing their significance in minimizing a divergence
quantity between two operators at different times
\cite{tsang22b}. Section~\ref{sec_prop} presents some fundamental
properties of the GCEs that are key to their applications to quantum
metrology, including a chain rule (Theorem~\ref{thm_chain}) that gives
the GCEs a Markovian property for a sequence of channels and a
Pythagorean theorem (Theorem~\ref{thm_pyth}) that gives the divergence
an additive property. Sections~\ref{sec_bayes} and \ref{sec_qrbt}
present the core results of this work, namely, the applications of a
version of the GCEs to quantum parameter estimation. This GCE follows
a particular operator ordering based on the Jordan product and is
shown to play a natural role in quantum estimation theory.

Section~\ref{sec_bayes} studies the role of the GCE in Bayesian
quantum parameter estimation, a topic that has received renewed
interest in recent years \cite{twc,macieszczak14,rubio19}.  Within
Sec.~\ref{sec_bayes}, Sec.~\ref{sec_bayes_gen} presents the general
relations between the Personick estimators for a sensor under
decoherence, Sec.~\ref{sec_dp} shows how they enable the use of
dynamic programming in quantum sensor measurement design, and
Sec.~\ref{sec_weak} discusses the special case of the real weak value.

Section~\ref{sec_qrbt} switches to the frequentist setting and
presents the quantum Rao-Blackwell theorem, Theorem~\ref{thm_qrbt}, in
Sec.~\ref{sec_qrbt_gen}.  Sections~\ref{sec_tensor}--\ref{sec_sum}
present some significant consequences of the quantum theorem for
quantum metrology, while Sec.~\ref{sec_therm} discusses an application
of the theorem to thermal-light sensing.

%Section~\ref{sec_qrbt} is independent of Sec.~\ref{sec_bayes} and one
%can go from Secs.~\ref{sec_review} and \ref{sec_prop} directly to
%Sec.~\ref{sec_qrbt}, skipping Sec.~\ref{sec_bayes} if one wishes.

Section~\ref{sec_con} is the conclusion, listing some open problems.
Appendix~\ref{app_predict} discusses the complementary concept of
  quantum prediction. Appendix~\ref{app_ce} reviews the classical
  conditional expectation to give the quantum formalism a more
  familiar context.  Appendix~\ref{app_gauss} gives an explicit
formula for the GCE for Gaussian systems. Appendix~\ref{app_vn}
  defines the von Neumann measurement. Appendix~\ref{app_dp} presents
  the dynamic-programming algorithm.  Appendix~\ref{app_rb} justifies
  the name of Theorem~\ref{thm_qrbt} by deriving the classical
  Rao-Blackwell theorem from it. Appendix~\ref{app_compare}
discusses the differences and relations between the Bayesian and
frequentist settings. Appendix~\ref{app_sinha} compares this work with
some prior works. Appendix~\ref{app_U} offers an alternative
derivation of the quantum U-statistics, first introduced by \guta{}
and Butucea \cite{guta10}, using Theorem~\ref{thm_qrbt}.
Appendix~\ref{app_proofs} contains the more technical proofs.

\section{\label{sec_review}Review of generalized conditional
  expectations}
This section follows Ref.~\cite{tsang22b} and Chap.~6 in
Ref.~\cite{hayashi}. Let $\mc O(\mc H)$ be the space of bounded
operators on a Hilbert space $\mc H$ and $\rho \in \mc O(\mc H)$ be a
density operator.  Define an inner product between two operators
$A,B \in \mc O(\mc H)$ and a norm as
\begin{align}
\Avg{B,A}_\rho &\equiv \trace B^\dagger \mc E_\rho A,
&
\norm{A}_\rho &\equiv \sqrt{\avg{A,A}_\rho},
\end{align}
where $\mc E_\rho:\mc O(\mc H)\to \mc O(\mc H)$ is a linear,
self-adjoint, and positive-semidefinite map with respect to the
Hilbert-Schmidt inner product
\begin{align}
\Avg{B,A}_{\HS} \equiv \trace B^\dagger A.
\end{align}
The weighted inner product $\avg{\cdot,\cdot}_\rho$ is a
generalization of the inner product between two random variables in
classical probability theory \cite{parth05}. Some desirable properties
of $\mc E$ are\footnote{ Equation~(\ref{ax3}) without the commutation
  condition is proposed in Ref.~\cite{hayashi} and repeated in
  Ref.~\cite{tsang22b}, but it turns out to be false for many operator
  products, including the Jordan product, as pointed out by
  Ref.~\cite{parzygnat23}. The results in Ref.~\cite{tsang22b} remain
  correct if the commutation condition in Eq.~(\ref{ax3}) is imposed.}
\begin{align}
\mc E_\rho A &= \rho A \textrm{ if $\rho,A$ commute},
\label{ax1}
\\
\mc E_\rho(U^\dagger A U) &= U^\dagger \bk{\mc E_{U\rho U^\dagger}A} U,
\label{ax2}
\\
\mc E_{\rho_1\otimes \rho_2}(A_1 \otimes A_2)
&= (\mc E_{\rho_1} A_1) \otimes (\mc E_{\rho_2}A_2)
\nonumber\\
\textrm{if $\rho_1,A_1$ commute} & 
\textrm{ or $\rho_2,A_2$ commute,}
\label{ax3}
\\
\norm{A_1\otimes I_2}_{\rho} &\le \norm{A_1}_{\trace_2\rho},
\label{ax4}
\end{align}
where $A$ is any operator on $\mc H$, $U$ is any unitary operator on
$\mc H$, $\rho_j$ is any density operator on $\mc H_j$, $\mc H_j$ is
any Hilbert space, $A_j$ is any operator on $\mc H_j$, $I_j$ is the
identity operator on $\mc H_j$, $\rho$ in Eq.~(\ref{ax4}) is any
density operator on $\mc H_1\otimes \mc H_2$, and $\trace_j$ denotes
the partial trace with respect to $\mc H_j$.  Examples of $\mc E$ that
satisfy Eqs.~(\ref{ax1})--(\ref{ax4}) include
\begin{align}
\mc E_\rho A &= \frac{1}{2}\bk{\rho A + A\rho},
\label{jordan}
\\
\mc E_\rho A &= \rho A,\\
\mc E_\rho A &= \sqrt{\rho} A \sqrt{\rho}.
\label{root}
\end{align}
In the following, I fix $\mc E$ to be a map that satisfies
Eqs.~(\ref{ax1})--(\ref{ax4}).

Let $L_2(\rho)$ be the completion of $\mc O(\mc H)$ with respect to
the norm $\norm{\cdot}_\rho$, such that it becomes a weighted Hilbert
space for the operators. Each element of $L_2(\rho)$ is then an
equivalence class of operators with zero distance between them. If
$\mc H$ is infinite-dimensional, $\mc O(\mc H)$ may not be complete
and $L_2(\rho)$ may include unbounded operators as well
\cite{holevo11}.  The infinite-dimensional case is much more
complicated to treat with rigor, so I consider only finite-dimensional
Hilbert spaces in the following for simplicity, and assume that the
results still hold for a couple of the infinite-dimensional problems
studied later in Appendix~\ref{app_gauss} and Sec.~\ref{sec_therm}.

\begin{definition}
\label{def_div}
Let $\sigma$ be a density operator on $\mc H_1$ and
$\mc F:\mc O(\mc H_1) \to \mc O(\mc H_2)$ be a completely positive,
trace preserving (CPTP) map that models a quantum channel.  Then the
divergence between an operator $A \in L_2(\sigma)$ and another
operator $B \in L_2(\mc F\sigma)$ is defined as
\begin{align}
D_{\sigma,\mc F}(A,B) &\equiv \norm{A}_\sigma^2 - 
2 \Re \Avg{\mc F^* B,A}_\sigma
+\norm{B}_{\mc F\sigma}^2,
\label{D}
\end{align}
where $\Re$ denotes the real part and $\mc F^*$ denotes the
Hilbert-Schmidt adjoint of $\mc F$.
\end{definition}
This divergence, introduced in Ref.~\cite{tsang22b}, can be related to
the more usual definition of distance in a larger Hilbert space by
considering the Stinespring representation
\begin{align}
\mc F \sigma = \trace_{10} U(\sigma\otimes\tau) U^\dagger,
\label{stinespring}
\end{align}
where $\tau$ is a density operator on $\mc H_2\otimes\mc H_0$,
$\mc H_0$ is some auxiliary Hilbert space, $U$ is a unitary operator
on $\mc H_1\otimes\mc H_2\otimes \mc H_0$ that models the evolution
from time $t$ to time $T \ge t$, and $\trace_{10}$ is the partial
  trace over $\mc H_1$ and $\mc H_0$.  Let $\rho = \sigma\otimes\tau$
and define the Heisenberg pictures of $A$ and $B$ as
\begin{align}
A_t &\equiv A\otimes I_2\otimes I_0,
&
B_T &\equiv U^\dagger (I_1\otimes B\otimes I_0)U.
\label{hei}
\end{align}
Then it can be shown that
\begin{align}
D_{\sigma,\mc F}(A,B) &\ge \norm{A_t - B_T}_{\rho}^2,
\label{D_dist}
\end{align}
and the divergence is nonnegative.  Furthermore, if the $\mc E$ map
obeys the stricter equality condition in Eq.~(\ref{ax4}), then the
equality in Eq.~(\ref{D_dist}) holds, and $D$ is exactly the squared
distance in the larger Hilbert space.

\begin{definition}
\label{def_gce}
Given a density operator $\sigma$, a CPTP map $\mc F$, and an
  $\mc E$ map that satisfies Eqs.~(\ref{ax1})--(\ref{ax4}), the GCE
$\gce{\mc F} :L_2(\sigma) \to L_2(\mc F\sigma)$ of $A \in L_2(\sigma)$
is defined as 
\begin{align}
\gce{\mc F} A &\equiv \argmin_{B \in
      L_2(\mc F\sigma)} D_{\sigma,\mc F}(A,B),
\label{gce_def}
\end{align}
which leads to 
\begin{align}
\Avg{c,\gce{\mc F} A}_{\mc F\sigma}
&= \Avg{\mc F^* c, A}_{\sigma}
\quad
\forall c \in L_2(\mc F\sigma).
\label{gce}
\end{align}
More explicitly, $\gce{\mc F} A$ is an equivalence class of operators
that satisfy
\begin{align}
\mc E_{\mc F\sigma} \gce{\mc F} A &= \mc F \mc E_\sigma A.
\label{gce_exp}
\end{align}
\end{definition}
Equation~(\ref{gce}) can be derived by assuming the ansatz
$B = \gce{\mc F} A + \epsilon c$ with $\epsilon \in \mathbb R$,
$c \in L_2(\mc F\sigma)$, and minimizing $D$ with respect to
$\epsilon$. Given an $A$, the existence and uniqueness of
$\gce{\mc F} A$ as an element of $L_2(\mc F\sigma)$ can be proved by
viewing Eq.~(\ref{gce}) as a linear functional of $c$ and applying the
Riesz representation theorem \cite{debnath05}.
Equation~(\ref{gce_exp}) can also be derived independently from a
state-over-time formalism \cite{parzygnat23}.  With the GCE, the
minimum divergence becomes
\begin{align}
D_{\sigma,\mc F}(A,\gce{\mc F} A)
&= \min_{B \in L_2(\mc F\sigma)}D_{\sigma,\mc F}(A,B)
\nonumber\\
&= \norm{A}_\sigma^2 - \norm{\gce{\mc F} A}_{\mc F\sigma}^2.
\label{Dmin}
\end{align}

Note that the GCE map $\gce{\mc F}$ depends implicitly on the $\mc E$
map and the prior state $\sigma$; the choice of $\mc E$ and $\sigma$
should be clear from the context in the following and, when necessary,
$\sigma$ is stated explicitly in the superscript as
$\gce{\mc F}^\sigma$. Note also that Chap.~6 in Ref.~\cite{hayashi}
writes $\gce{\mc F}^\sigma$ as $\mc F_{\sigma,x}$, where $x$ denotes
the $\mc E$ map being used, while Ref.~\cite{tsang22b} writes
$\gce{\mc F}^\sigma$ as $\mc F_\sigma$.  Appendix~\ref{app_predict}
presents more interesting formulas concerning $\mc F^*$ and
$\gce{\mc F}$ that justify the new notations, while
Appendix~\ref{app_ce} presents a brief and elementary review of the
classical conditional expectation to give the quantum formalism a more
familiar context.

Some examples are in order. Consider the unitary channel
\begin{align}
\mc F\sigma &= U \sigma U^\dagger,
\label{unitary}
\end{align}
where $U$ is a unitary operator on $\mc H_1$. A solution to any GCE
is
\begin{align}
\gce{\mc F} A &= U A U^\dagger,
\label{gce_unitary}
\end{align}
leading to $D_{\sigma,\mc F}(A,\gce{\mc F} A) = 0$.
Equation~(\ref{gce_unitary}) is called the Heisenberg representation
in quantum computing \cite{gottesman99}, and the GCEs can be regarded
as generalizations of the Heisenberg representation for open systems.

With the root product given by Eq.~(\ref{root}), the GCE becomes the
Accardi-Cecchini GCE \cite{accardi82,petz}, and its Hilbert-Schmidt
adjoint $(\gce{\mc F})^*$ is known as the Petz recovery map,
which is useful in quantum information theory \cite{wilde17}.

Appendix~\ref{app_gauss} presents another example where $\sigma$ is a
Gaussian state, $\mc F$ is a Gaussian channel \cite{holevo19}, and $A$
is a quadrature operator. Then the GCE in terms of the Jordan product
given by Eq.~(\ref{jordan}) and the associated divergence turn out to
have the same formulas as the classical conditional expectation and
its mean-square error for the usual linear Gaussian model
\cite{anderson}.

\section{\label{sec_prop}Fundamental properties}
With Eqs.~(\ref{gce})--(\ref{Dmin}), it is straightforward to prove
the following crucial properties of the GCE:
\begin{theorem}[Chain rule\footnote{I follow Ref.~\cite{accardi82} to
    call this property a chain rule. Note that Ref.~\cite{hayashi}
    calls it associativity, while
    Refs.~\cite{parzygnat23,parzygnat23a} call it compositionality.};
  see Eq.~(6.22) in Ref.~\cite{hayashi}]
\label{thm_chain}
Let $\mc G:\mc O(\mc H_2) \to \mc O(\mc H_3)$ be another CPTP map.
Then the GCE of the composite map $\mc G\mc F$ is given by
\begin{align}
\gce{(\mc G \mc F)}^\sigma A  &= 
\gce{\mc G}^{\mc F\sigma} \gce{\mc F}^\sigma A,
\end{align}
which can be abbreviated as
\begin{align}
\gce{(\mc G \mc F)}  &= \gce{\mc G} \gce{\mc F}.
\end{align}
In other words, the GCE for a chain of CPTP maps is given by a chain
of the GCEs associated with the individual CPTP maps.
\end{theorem}

\begin{theorem}[Pythagorean theorem]
\label{thm_pyth}
  Given the two CPTP maps $\mc F$ and $\mc G$, the minimum divergences
  obey
\begin{align}
D_{\sigma,\mc G\mc F}(A,\gce{(\mc G\mc F)} A)
&= 
D_{\sigma,\mc F}(A,\gce{\mc F} A)
\nonumber\\&\quad
+ D_{\mc F\sigma,\mc G}(\gce{\mc F} A,\gce{\mc G}\gce{\mc F} A).
\end{align}
\end{theorem}
\begin{proof}
  Use Eq.~(\ref{Dmin}) and Theorem~\ref{thm_chain}.
\end{proof}
Figure~\ref{GCE_cat} offers some diagrams that illustrate the
theorems.

\begin{figure}[htbp!]
\centerline{\includegraphics[width=0.48\textwidth]{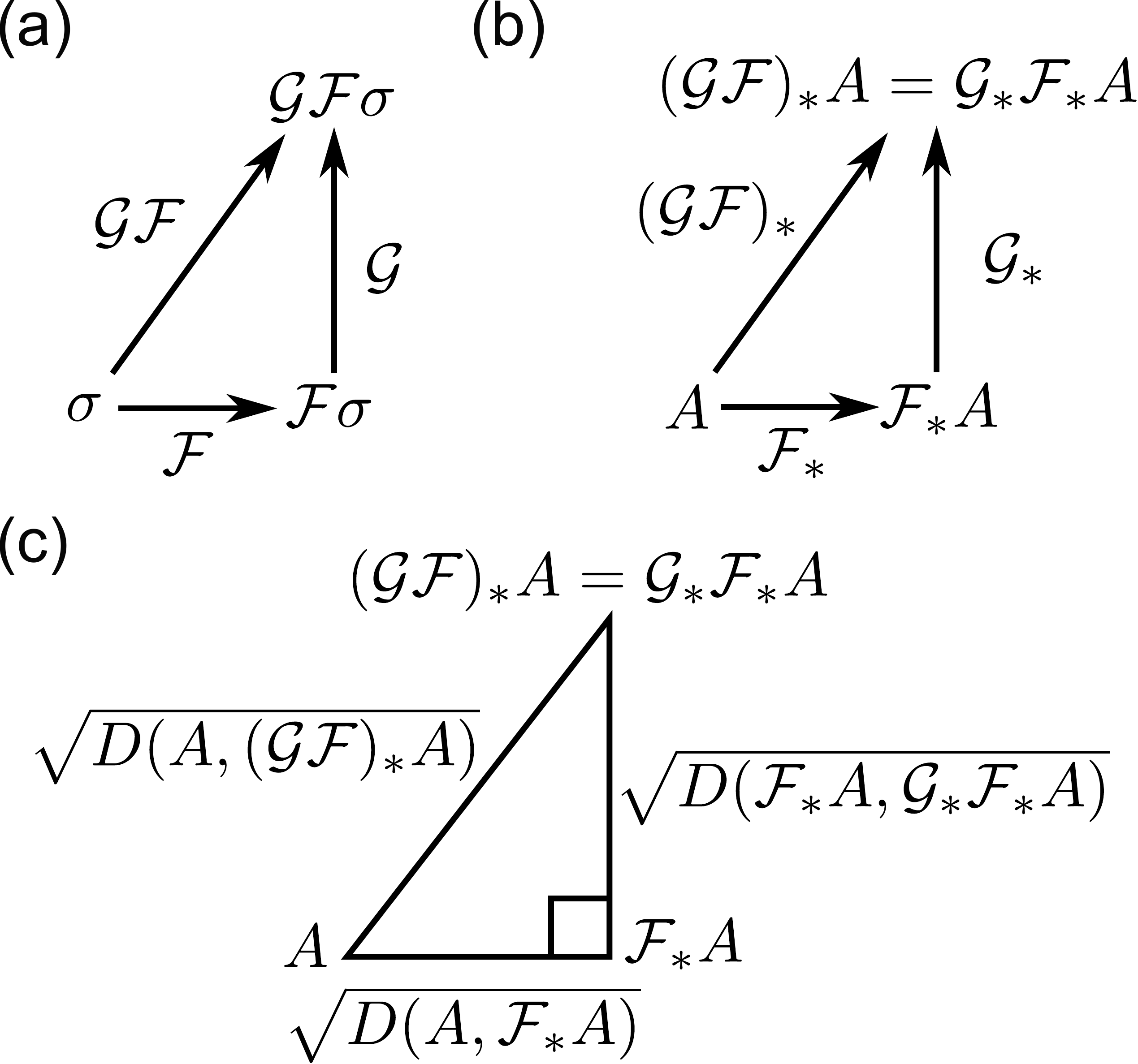}}
\caption{\label{GCE_cat} (a) A diagram depicting the map of a density
  operator $\sigma$ through the CPTP maps $\mc F$ and then $\mc G$.
  (b) A diagram depicting the map of an observable $A$ through the GCE
  $\gce{(\mc G\mc F)}$, or equivalently through the two GCEs
  $\gce{\mc F}$ and then $\gce{\mc G}$, as per
  Theorem~\ref{thm_chain}. (c) A diagram depicting the root
  divergences between the operators as lengths of the sides of a right
  triangle, as per Theorem~\ref{thm_pyth}.  The subscripts of $D$ are
  omitted for brevity.}
\end{figure}

Before moving on, I list two more properties of the GCEs---their
physical significance for generalizing the Rao-Blackwell theorem
\cite{lehmann98} will be explained in Sec.~\ref{sec_qrbt}.
\begin{lemma}[Law of total expectation]
\label{lem_tot}
For any $A \in L_2(\sigma)$,
\begin{align}
\trace \sigma A &= \Avg{I_1,A}_\sigma = \Avg{I_2,\gce{\mc F} A}_{\mc F\sigma}
= \trace (\mc F\sigma) (\gce{\mc F} A).
\label{mean}
\end{align}
\end{lemma}
\begin{lemma}
\label{lem_var}
Let $a$ be any complex number. Then
\begin{align}
\norm{A - a I_1}^2_\sigma
&= \norm{\gce{\mc F} A - a I_2}^2_{\mc F\sigma}
+ D_{\sigma,\mc F}(A,\gce{\mc F} A).
\label{RB}
\end{align}
\end{lemma}
See Appendix~\ref{app_proofs} for the proofs of Lemmas~\ref{lem_tot}
and \ref{lem_var}.

A map $\gce{\mc F}$ that satisfies Lemma~\ref{lem_tot} is also called
a coarse graining \cite{petz}. Whereas Petz's definition requires a
coarse graining to be completely positive, the GCEs here need not
be. If $a$ in Lemma~\ref{lem_var} is set as the mean given by
Eq.~(\ref{mean}), then Lemma~\ref{lem_var} says that the generalized
variance of $\gce{\mc F} A$ given by
$\norm{\gce{\mc F} A - a I_2}_{\mc F\sigma}^2$ cannot exceed that of
$A$.

The mathematics of GCEs would be uncontroversial if not for its
physical implication: By defining a divergence between two operators
at different times, a retrodiction of a hidden quantum observable $A$
can be given a risk measure and therefore a meaning in the spirit of
decision theory \cite{berger}. In other words, after a channel $\mc F$
is applied, one can seek an observable $B$ that is the closest to $A$
if the divergence is regarded as a squared distance, and
$\gce{\mc F} A$ is the answer. It remains an open and reasonable
question, however, why the divergence between two operators is an
important quantity. If $A_t$ at time $t$ does not commute with $B_T$
at a later time in the Heisenberg picture, where $A_t$ and $B_T$ are
defined by Eqs.~(\ref{hei}), then Belavkin's nondemolition principle
for their simultaneous measurability is violated
\cite{belavkin_qnd,bouten,gough20}, no classical observer can access
the precise values of both, and the divergence does not seem to have
any obvious meaning to the classical world. Notably, Gough claims in
Ref.~\cite{gough20} that a retrodiction that violates the
nondemolition principle is ``misapplying Bayes theorem,'' ``not
possible,'' and ``unwarranted.''  Reference~\cite{gough19}, the
preprint version of Ref.~\cite{gough20}, goes even further in claiming
that someone who does not follow the principle may obtain ``wholly
meaningless'' answers and is ``in a state of sin.''  In
Ref.~\cite{james21}, James also claims that the principle should be
observed for a quantum conditional expectation to make sense.  To show
that a retrodiction can make sense beyond the nondemolition principle,
the next sections offer natural scenarios in quantum metrology that
will give operational meanings to a GCE and the associated divergence.

\section{\label{sec_bayes}Bayesian quantum parameter estimation}

\subsection{\label{sec_bayes_gen}General results}
Consider the typical setup of Bayesian quantum parameter estimation
\cite{personick71} depicted in Fig.~\ref{bayes_metrology}(a).  Let $X$
be a hidden classical random parameter with a countable parameter
space $\mc X$ and a prior probability distribution
$P_X:\mc X \to [0,1]$. A quantum sensor is coupled to $X$, such that
its density operator conditioned on $X = x$ is
$\rho_x \in \mc O(\mc H_2)$. A classical observer measures the quantum
sensor, as modeled by a positive operator-valued measure (POVM)
$M:\Sigma_{\mc Y} \to \mc O(\mc H_2)$ on a Borel space
$(\mc Y,\Sigma_{\mc Y})$, where $\Sigma_{\mc Y}$ is the Borel
sigma-algebra of $\mc Y$ \cite{holevo11}.  The observer uses the
outcome $y \in \mc Y$ to estimate the value of a real random variable
$a:\mc X \to \mathbb R$. The problem can be framed in the GCE
formalism by writing
\begin{align}
\sigma &= \sum_x P_X(x)\ket{x}\bra{x},
\quad
A = \sum_x  a(x) \ket{x}\bra{x},
\label{bayes1}
\\
\mc F\sigma &= \sum_x  \rho_x \bra{x}\sigma\ket{x} = \sum_x \rho_x P_X(x),
\label{bayes2}
\end{align}
where $\{\ket{x}:x \in \mc X\}$ is an orthonormal basis of $\mc H_1$
and the classical random variable $a(X)$ is framed as the hidden
  observable $A$ discussed in Secs.~\ref{sec_review} and
  \ref{sec_prop}.  $\mc F$ here is called a classical-quantum channel
and has a natural generalization in the infinite-dimensional case
\cite{holevo19}.

\begin{figure}[htbp!]
\centerline{\includegraphics[width=0.48\textwidth]{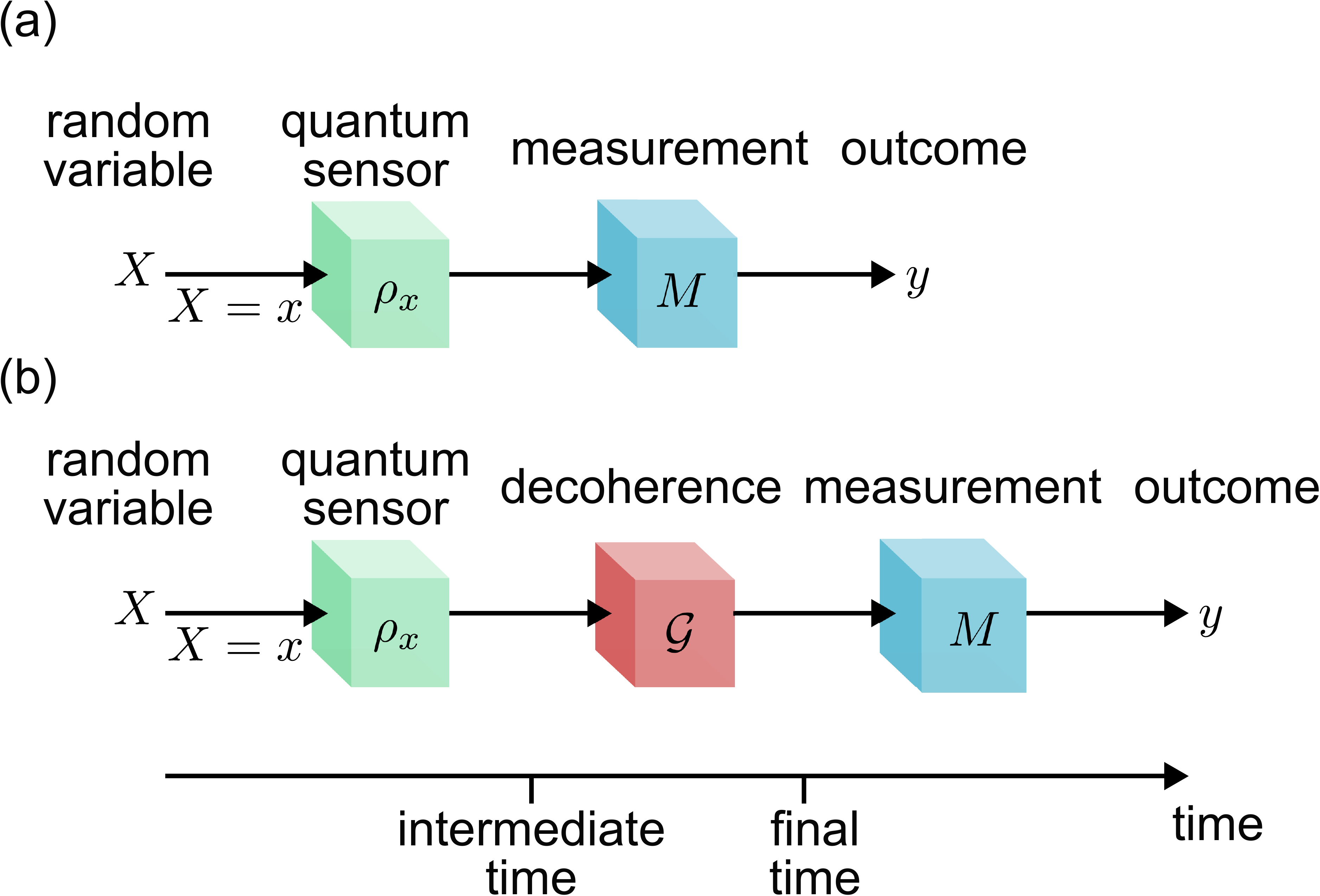}}
\caption{\label{bayes_metrology}Some scenarios of Bayesian quantum
  parameter estimation. See the main text for the definitions of the
  symbols.}
\end{figure}

In the following, I consider only Hermitian operators (observables)
and assume $\mc E$ to be the Jordan product given by
Eq.~(\ref{jordan}), such that all the operator Hilbert spaces are
real, the equalities in Eqs.~(\ref{ax4}) and (\ref{D_dist}) hold, and
the GCE is in fact a projection in the larger Hilbert space
\cite{hayashi}. 

Suppose that a von Neumann measurement of an observable $B$ on
$\mc H_2$ is performed, as defined in Appendix~\ref{app_vn}, and
the outcome is used as the estimator. I call such a $B$ an
  operator-valued estimator.  The mean-square estimation error
averaged over the prior is given by
\begin{align}
  \sum_{x} P_X(x)\int [b - a(x)]^2 \trace \Pi(db) \rho_x
  &= D_{\sigma,\mc F}(A,B),
\label{MSE_bayes}
\end{align}
where $\Pi$ is the projection-valued measure of $B$.
Equation~(\ref{MSE_bayes}) is precisely the divergence in
Definition~\ref{def_div}. According to the seminal work of Personick
\cite{personick71}, the optimal operator-valued estimator is the GCE
$\gce{\mc F} A$, and the minimum error, hereafter called the Bayesian
error, is $D_{\sigma,\mc F}(A,\gce{\mc F} A)$. It can also be shown
that the von Neumann measurement of $\gce{\mc F} A$ remains optimal
even if POVMs are considered (see Sec.~VIII~1(d) in
Ref.~\cite{helstrom}, Appendix~A in Ref.~\cite{macieszczak14}, or
Corollary~\ref{cor_vn} below).

Now suppose that a complication occurs in the experiment, as depicted
by Fig.~\ref{bayes_metrology}(b): Before the measurement can be
performed, the sensor is further corrupted by decoherence, as modeled
by another CPTP map $\mc G$. The error of an operator-valued
  estimator $B'$ is now
\begin{align}
\sum_x P_X(x) 
\int [b - a(x)]^2 \trace \Pi'(db) \mc G \rho_x
&= D_{\sigma,\mc G\mc F}(A,B'),
\label{MSE_bayes2}
\end{align}
where $\Pi'$ is the projection-valued measure of $B'$. The Personick
estimator after $\mc G$ is then $B'=\gce{(\mc G\mc F)} A$, and the
Bayesian error becomes
$D_{\sigma,\mc G\mc F}(A,\gce{(\mc G\mc F)} A)$.  A fundamental fact
is as follows.
\begin{corollary}[Monotonicity of the Bayesian error]
\label{cor_mono}
The Bayesian error cannot decrease under decoherence, viz.,
\begin{align}
D_{\sigma,\mc G\mc F}(A,\gce{(\mc G\mc F)} A)
\ge 
D_{\sigma,\mc F}(A,\gce{\mc F} A)
\label{mono}
\end{align}
for the estimation problem modeled by
Eqs.~(\ref{bayes1})--(\ref{MSE_bayes2}).
\end{corollary}
\begin{proof}
  Use Theorem~\ref{thm_pyth} and the nonnegativity of $D$.
\end{proof}
The scenario so far is standard and uncontroversial, as $A$ is
effectively a classical random variable. Mathematically, $A_t$ and
$(\gce{\mc F} A)_T$ in the Heisenberg picture commute (see Sec.~IV~F
in Ref.~\cite{tsang22b}) and thus satisfy the nondemolition principle;
so do $A_t$ and $[\gce{(\mc G\mc F)} A]_T$. Physically, the principle
implies that another classical observer can, in theory, access the
precise value of $A$ in each trial, the estimates can be compared with
the true values by the classical observers after the trials, and $D$
is their expected error. The monotonicity given by
Corollary~\ref{cor_mono} is a noteworthy result, but unsurprising.

More can be said about the error increase, hereafter called the regret
(to borrow a term from decision theory \cite{berger}). First of all,
the chain rule in Theorem~\ref{thm_chain} gives an operational meaning
to the GCE $\gce{\mc G}$ as the map that relates the intermediate
Personick estimator $\gce{\mc F} A$ to the final
$\gce{(\mc G\mc F)} A = \gce{\mc G}\gce{\mc F} A$.  In other words,
the final Personick estimator is equivalent to a retrodiction of the
intermediate $\gce{\mc F} A$, which is a quantum observable.  Second,
the Pythagorean theorem in Theorem~\ref{thm_pyth} means that the
regret caused by the decoherence is precisely the divergence between
the intermediate and final estimators:
\begin{align}
D_{\sigma,\mc G\mc F}(A,\gce{(\mc G\mc F)} A)
- D_{\sigma,\mc F}(A,\gce{\mc F} A)
\nonumber\\
= D_{\mc F\sigma,\mc G}(\gce{\mc F} A,\gce{\mc G}\gce{\mc F} A).
\label{penalty}
\end{align}
The two divergences on the left-hand side have a firm
decision-theoretic meaning as estimation errors because $A$ is
classical.  It follows that, even though the divergence on the
right-hand side is between two quantum observables, it also has a firm
decision-theoretic meaning as the regret---for being unable to perform
the optimal measurement and having to suffer from the decoherence.  As
the regret concerns the performances of the two estimators in separate
experiments, the estimators need not obey the nondemolition
  principle, which is a condition on two observables to be
  simultaneously measurable in the same experiment.

I stress that the regret is not a contrived concept invented here
solely to give an operational meaning to the divergence---its
classical version is an established concept in information theory and
Bayesian learning \cite{weissman10,wu11,mismatch,xu22}.

\subsection{\label{sec_dp}Dynamic programming}
When the decoherence is modeled by a chain of CPTP maps
$\mc G = \mc F_N\dots \mc F_2$, the final error is the sum of
all the incremental regrets along the way, viz.,
\begin{align}
D_{\sigma,\mc G_N}(A,\mc G_{N*} A)
&= \sum_{n=1}^N D_n,
\label{add}
\\
\mc G_n &\equiv \mc F_n\dots \mc F_2 \mc F_1,
\label{Gn}
\\
D_n &\equiv D_{\sigma_n,\mc F_n}(A_n,A_{n+1}),
\\
\sigma_{n+1} &\equiv \mc G_n\sigma = \mc F_n \sigma_n,
\quad
\sigma_1 = \sigma,
\label{sigma_n}
\\
A_{n+1} &\equiv \mc G_{n*} A = 
\mc F_{n*} A_n,
\quad
A_1 = A,
\label{A_n}
\end{align}
where $\mc F_1 = \mc F$ for the parameter estimation problem, so
even the error at the first step
$D_1 = D_{\sigma,\mc F}(A,\gce{\mc F} A)$ can be regarded as a
regret. Every $D_n$, bar $D_1$, is a divergence between a
quantum observable $A_n$ and its estimator $A_{n+1}$ that may
not commute in the Heisenberg picture.

Suppose that the experimenter can choose the channels
$(\mc F_1,\dots,\mc F_N)$ from a set of options and would like to find
the optimal choice that minimizes the final error. One example
  is the use of a programmable photonic circuit \cite{bogaerts20} to
  measure light for sensing or imaging. The Markovian nature of
Eqs.~(\ref{sigma_n})--(\ref{A_n}) and the additive nature of the final
error given by Eq.~(\ref{add})---which originate from
Theorems~\ref{thm_chain} and \ref{thm_pyth}---are precisely the
conditions that make this optimal control problem amenable to dynamic
programming \cite{bertsekas}, an algorithm that can reduce the
computational complexity substantially \cite{cormen}. To be specific,
let the system state (in the context of control theory) at time $n$ be
$s_n \equiv (\sigma_n,A_n)$. Then Eqs.~(\ref{add})--(\ref{A_n}) imply
that the state dynamics and the final error can be expressed as
\begin{align}
s_{n+1} &= f(s_{n},\mc F_n),
\label{markov}
\\
D_{\sigma,\mc G_N}&= \sum_{n=1}^N g(s_{n},\mc F_n)
\label{Dsum}
\end{align}
in terms of some functions $f$ and $g$.  Equations~(\ref{markov}) and
(\ref{Dsum}) are now in the form of a Markov decision process that is
amenable to dynamic programming for computing the optimal maps
$(\mc F_1,\dots,\mc F_N)$ among the set of options to minimize
the final error \cite{bertsekas}; Appendix~\ref{app_dp} describes
  the algorithm for the reader's information.  As dynamic programming
is a cornerstone of control theory, there exist a plethora of exact or
approximate methods to implement it, such as neural networks under the
guise of reinforcement learning \cite{bertsekas_rl}.

\subsection{\label{sec_weak}Weak value}
To elaborate on the operational meaning for the weak value, which is a
GCE of a quantum observable given a prior state and a measurement
outcome \cite{ohki15,ohki18,dressel15,tsang22b}, let us return to the
scenario depicted by Fig.~\ref{bayes_metrology}(a). Suppose, for
mathematical simplicity, that the outcome space $\mc Y$ is
countable. The measurement can be framed as a $\mc G$ map given by
\begin{align}
\mc G \tau = \sum_y [\trace M(y) \tau]\ket{y}\bra{y},
\label{G_measure}
\end{align}
where $\{\ket{y}:y \in \mc Y\}$ is an orthonormal basis of $\mc H_3$
and $M$ is the POVM of a measurement that may not be optimal. An
estimator $b:\mc Y \to \mathbb R$ as a function of the measurement
outcome can be framed as the observable
\begin{align}
B = \sum_y b(y)\ket{y}\bra{y}.
\label{B}
\end{align}
The GCE then leads to the optimal estimator
\begin{align}
B &= \gce{\mc G}\gce{\mc F} A,
\\
b(y) &= \frac{\trace M(y) \mc E_{\mc F\sigma}\gce{\mc F} A}
{\trace M(y) \mc F\sigma}
= \frac{\Re \trace M(y) (\gce{\mc F} A) (\mc F\sigma) }
{\trace M(y) \mc F\sigma}.
\label{wv}
\end{align}
The last expression in Eq.~(\ref{wv}) is precisely the definition
  of the real weak value of $\gce{\mc F} A$ given a prior
  state $\mc F\sigma$ and a measurement outcome $y$ (see, for example,
  Eq.~(3.13) in Ref.~\cite{wiseman02}, Eq.~(10) in Ref.~\cite{hall04},
  or Eq.~(5) in Ref.~\cite{dressel15}). Moreover, the divergence
between the ideal $\gce{\mc F} A$ and the $B$ associated with
the weak value is precisely the regret caused by the suboptimality of
the measurement $M$, as per Theorem~\ref{thm_pyth}.  Hence, regardless
of how anomalous the weak value may seem, it does have an operational
role in parameter estimation, and its divergence from the ideal
$\gce{\mc F} A$ has a concrete decision-theoretic meaning as
the regret for not using the optimal measurement.

The preceding discussion also serves as a rough proof of the following
corollary, which is proved by different methods in Sec.~VIII~1(d) of
Ref.~\cite{helstrom} and Appendix~A of Ref.~\cite{macieszczak14}.
\begin{corollary}
\label{cor_vn}
No POVM can improve upon the Bayesian error
$D_{\sigma,\mc F}(A,\gce{\mc F} A)$ achieved by a von Neumann
measurement of $\gce{\mc F} A$.
\end{corollary}
Corollary~\ref{cor_vn} may be regarded as a consequence of
monotonicity, since any measurement with a countable set of outcomes
can be framed as a CPTP $\mc G$ map given by Eq.~(\ref{G_measure}),
and by Corollary~\ref{cor_mono}, the error cannot decrease. A POVM
with a more general outcome space can still be framed as a
quantum-classical channel; see, for example, Theorem~2 in
Ref.~\cite{barchielli06}, but it requires a mathematical framework far
more complex than what is necessary for this work. An easier proof for
general POVMs, to be presented in Appendix~\ref{app_proofs}, is to use
a later result in Sec.~\ref{sec_qrbt}.

Note that the optimality of the weak value here does not contradict
Ref.~\cite{ferrie_combes}, which shows that weak-value amplification,
a procedure that involves postselection (i.e., discarding some of the
outcomes), is suboptimal for metrology. Here, the weak value given by
Eq.~(\ref{wv}) is used directly as an estimator with any measurement
outcome, and no postselection is involved. Note also that the
optimality is in the specific context of finding the best estimator
after a given measurement; it does not mean that any measurement
method that is heuristically inspired by the weak-value concept, such
as weak-value amplification, can be optimal.  In fact, by virtue of
Corollary~\ref{cor_vn}, such methods can never outperform the optimal
von Neumann measurement.

\section{\label{sec_qrbt}A quantum Rao-Blackwell theorem}
\subsection{\label{sec_qrbt_gen}General result}
In classical frequentist statistics, the Rao-Blackwell theorem is
among the most useful applications of the conditional expectation
\cite{lehmann98,thompson12,robert}. Here I outline a quantum
generalization.  Suppose that the quantum sensor is modeled by a
family of density operators
$\{\rho_x:x \in \mc X\} \subset \mc O(\mc H_2)$, where the unknown
parameter $x$ is now deterministic and there is no longer any need to
assume a countable parameter space $\mc X$. A parameter of interest
$a:\mc X \to \mathbb R$ is to be estimated by a Hermitian
operator-valued estimator $B \in L_2(\rho_x)$, which need not be
unbiased or optimal in any sense. The local mean-square error (MSE)
upon a von Neumann measurement of $B$, as a function of $x \in \mc X$
and without being averaged over any prior, is given by
\begin{align}
\MSE_x &\equiv 
\int [b-a(x)]^2 \trace \Pi(db)\rho_x = \norm{B - a(x)I_2}_{\rho_x}^2,
\label{MSE}
\end{align}
where $\Pi$ is the projection-valued measure of $B$ and the Jordan
product is again assumed for $\mc E$. Without the unbiased condition
$\trace \rho_x B = a(x)$, the quantum Cram\'er-Rao bounds commonly
used in quantum metrology \cite{helstrom,hayashi} do not apply to
$\MSE_x$. There is no longer any simple optimality criterion for an
estimator in this general setting, but one can still construct a
\emph{partial order} of preference between estimators. Following
classical statistics (see p.~48 in Ref.~\cite{lehmann98}), I say that
an operator-valued estimator $B'$ dominates another estimator $B$ if
the error $\MSE_x'$ of the former never exceeds the error $\MSE_x$ of
the latter for all $x \in \mc X$ and can go strictly lower for some
$x$. I call an estimator admissible if it is dominated by none.

In classical frequentist statistics, there can be many admissible
estimators for one problem with no clear winner among them, and it may
be hard to even prove that a given estimator is admissible. The
Rao-Blackwell theorem is then a valuable tool for improving estimators
or for proofs regarding admissibility. A quantum version of the
theorem can be similarly useful for quantum estimation problems.

To state the quantum theorem, suppose that the quantum sensor goes
through a channel modeled by a CPTP map
$\mc G:\mc O(\mc H_2)\to \mc O(\mc H_3)$ and the GCE $\gce{\mc G}B$ in
terms of $\rho_x$ and the Jordan product is used as an estimator. The
error becomes
\begin{align}
\MSE_x' &= \norm{\gce{\mc G} B - a(x)I_3}_{\mc G\rho_x}^2.
\label{MSE2}
\end{align}
Lemma~\ref{lem_var} can now be used to prove the following.

\begin{theorem}[Quantum Rao-Blackwell theorem]
\label{thm_qrbt}
Let $\{\rho_x:x \in \mc X\}$ be a family of density operators,
$a:\mc X \to \mathbb R$ be an unknown parameter, $B$ be a Hermitian
operator-valued estimator, and $\MSE_x$ be the local error at
$x \in \mc X$. If a channel $\mc G$ is applied and the GCE
  $\gce{\mc G}B$ in terms of $\rho_x$ and the Jordan product
  does not depend on $x$, then the error $\MSE_x'$ of
$\gce{\mc G}B$ as an estimator is lower by the amount
\begin{align}
\MSE_x - \MSE_x' &=  D_{\rho_x,\mc G}(B,\gce{\mc G}B).
\label{qrbt}
\end{align}
\end{theorem}
\begin{proof}
Subtract Eq.~(\ref{MSE2}) from Eq.~(\ref{MSE}) and apply Lemma~\ref{lem_var}.
\end{proof}
For $\mc G$ to be realizable and $\gce{\mc G}B$ to be a valid
estimator, both cannot depend on the unknown $x$. When there are many
operator solutions to $\gce{\mc G}B$ that satisfy
Definition~\ref{def_gce}, any of the solutions can be the estimator in
Theorem~\ref{thm_qrbt} as long as it does not depend on $x$.

To demonstrate that Theorem~\ref{thm_qrbt} is indeed a quantum
  generalization of the Rao-Blackwell theorem, Appendix~\ref{app_rb}
  derives the classical theorem from Theorem~\ref{thm_qrbt}. As also
  shown in Appendix~\ref{app_rb}, a parameter-independent conditional
expectation in classical statistics can be obtained by conditioning on
a sufficient statistic.  The conditional expectation can then be used
to improve an estimator in a process called Rao-Blackwellization
\cite{lehmann98}.  Roughly speaking, Rao-Blackwellization works by
averaging the estimator with respect to unnecessary parts of the data,
thereby reducing its variance. A quantum Rao-Blackwellization, enabled
by Theorem~\ref{thm_qrbt}, can be similarly useful for improving a
quantum measurement if one can find a channel $\mc G$ that satisfies
the constant GCE condition and gives a large divergence between $B$
and $\gce{\mc G}B$. As long as
  $D_{\rho_x,\mc G}(B,\gce{\mc G}B) > 0$ for some $x$, the
  Rao-Blackwell estimator $\gce{\mc G}B$ dominates the original
  estimator $B$. The improvement stems from two basic facts about the
GCE: $\gce{\mc G}B$ maintains the same bias as that of $B$ by virtue
of Lemma~\ref{lem_tot}, while the variance of $\gce{\mc G}B$ cannot
exceed that of $B$ by virtue of Lemma~\ref{lem_var}. Roughly speaking,
the quantum Rao-Blackwell theorem works in the same way as the
classical case by averaging the estimator with respect to unnecessary
degrees of freedom via the GCE, thereby reducing its variance.

For the confused readers who wonder how a channel increases the error
in the Bayesian setting because of monotonicity but reduces the error
in the frequentist setting because of the Rao-Blackwell theorem,
Appendix~\ref{app_compare} offers a clarification.

It is noteworthy that Sinha also proposed some quantum Rao-Blackwell
theorems recently \cite{sinha22}, although his versions impose
stringent conditions on the commutativity of the operators.  Another
relevant prior work is Ref.~\cite{luczak15} by \L{}uczak, which
studies a concept of sufficiency in von Neumann algebra for
minimum-variance unbiased estimation but also makes some stringent
assumptions. These prior works, while seminal and mathematically
impressive, have questionable relevance to quantum metrology and are
discussed in more detail in Appendix~\ref{app_sinha}.

Given the close relation between the Rao-Blackwell theorem and the
concept of sufficient statistics in the classical case, it is natural
to wonder if a similar relation exists between the quantum
Rao-Blackwell theorem here and the concept of sufficient channels
defined by Petz \cite{petz}. One equivalent condition for a channel
$\mc G$ to be sufficient in Petz's definition is that the
Accardi-Cecchini GCE $\gce{\mc G}$ in terms of the root product
given by Eq.~(\ref{root}) does not depend on $x$. The GCE here, on the
other hand, is in terms of the Jordan product so that it can be
related to the parameter estimation error. The relation between Petz's
sufficiency and the constant GCE condition desired here is thus
nontrivial.

A trivial example that makes any GCE constant and the channel
sufficient in any sense is the unitary channel given by
Eqs.~(\ref{unitary}) and (\ref{gce_unitary}), as long as the unitary
operator there does not depend on $x$. Applying Theorem~\ref{thm_qrbt}
to the unitary channel gives no error reduction, however. In the
following, I offer more useful examples that both satisfy Petz's
sufficiency and give the desired constant GCE condition.

\subsection{\label{sec_tensor}A sufficient channel for tensor-product states}
\begin{lemma}
\label{lem_sep}
Let 
\begin{align}
\rho_x &= \sigma_x \otimes \tau,
&
\mc G\rho_x &= \trace_0  \rho_x = \sigma_x,
\label{G_sep}
\end{align}
where $\sigma_x$ is a density operator on $\mc H_1$ and $\tau$ is an
auxiliary density operator on $\mc H_0$. A solution to any GCE is
\begin{align}
\gce{\mc G}B = \trace_0[(I_1 \otimes \tau) B ],
\label{gce_sep}
\end{align}
% which coincides with Umegaki's conditional expectation (see
% Example~9.6 in Ref.~\cite{petz})
which does not depend on $x$ if $\tau$ does not.
\end{lemma}
See Appendix~\ref{app_proofs} for the proof.

A sufficient channel may be understood intuitively as a channel that
retains all information in the quantum sensor about $x$. Then it makes
sense that the channel in Lemma~\ref{lem_sep} is sufficient, as it
simply amounts to discarding an independent ancilla that carries no
information about $x$.  A significant implication of the lemma is a
more general version of Corollary~\ref{cor_vn} for the local error as
follows.
\begin{corollary}
\label{cor_vn2}
Given any POVM $M:\Sigma_{\mc Y} \to \mc O(\mc H_2)$, any estimator
$b:\mc Y \to \mathbb R$, and the resulting local error $\MSE_x$, there
exists a von Neumann measurement that can perform at least as well for
all $x \in \mc X$.
\end{corollary}
\begin{proof}
Write the Naimark extension of the POVM as
\begin{align}
\trace M(\cdot) \sigma_x &= \trace  \Pi'(\cdot)(\sigma_x \otimes \tau),
\end{align}
where $\sigma_x$ and $\tau$ are defined in Lemma~\ref{lem_sep} and
$\Pi'$ is a projection-valued measure on $\mc H_1\otimes \mc H_0$.
As shown in Appendix~\ref{app_vn}, the measurement and the data
processing by $b$ can be framed as a von Neumann measurement of
$B = \int b(y) \Pi'(dy)$ on the larger Hilbert space, such that its
error $\MSE_x$ with respect to $\rho_x = \sigma_x\otimes\tau$ is given
by Eq.~(\ref{MSE}).  Now assume the channel in Lemma~\ref{lem_sep}.  A
solution to $\gce{\mc G}B$ is given by Eq.~(\ref{gce_sep}), which does
not depend on $x$. It follows from Theorem~\ref{thm_qrbt} that the
error $\MSE_x'$ achieved by a von Neumann measurement of
$\gce{\mc G}B$ is at least as good as $\MSE_x$ for all $x \in \mc X$.
\end{proof}
Note that Corollary~\ref{cor_vn2} is more general than
Corollary~\ref{cor_vn}, since the former applies to the local errors
for all parameter values, not just the average errors in the Bayesian
case.  A proof of Corollary~\ref{cor_vn} using Corollary~\ref{cor_vn2}
is presented in Appendix~\ref{app_proofs}.

The corollaries imply that, in seeking an admissible estimator for
estimating a real scalar parameter under a mean-square-error
criterion, it is sufficient to consider only von Neumann measurements,
and randomization via an independent ancilla is not helpful in both
Bayesian and frequentist settings. For example, optical amplification
has been suggested to improve astronomical measurements
\cite{kellerer16}, but since optical amplification must involve an
independent ancilla \cite{haus,caves82}, Corollary~\ref{cor_vn2}
implies that there always exists a von Neumann measurement that
performs at least as well.

The corollaries are reminiscent of a well known result saying that a
von Neumann measurement of the so-called symmetric logarithmic
derivative (SLD) operator can saturate the quantum Cram\'er-Rao bound
(see Sec.~6.4 in Ref.~\cite{hayashi}). Note, however, that the bound
assumes unbiased estimators and the differentiability of $\rho_x$,
while the SLD measurement may be a function of the unknown parameter
and thus unrealizable. The corollaries here, on the other hand, are
much more general and conclusive, as they apply to arbitrary
estimators and arbitrary families of density operators, while the von
Neumann measurements they offer are all parameter-independent.

Of course, one is often forced to use an ancilla in practice, such as
the optical probe in atomic metrology \cite{itano93} or optomechanics
\cite{braginsky}.  Then the divergence offers a measure of regret in
both Bayesian and frequentist settings through Theorems~\ref{thm_pyth}
and \ref{thm_qrbt}. Consider atomic metrology for an example
\cite{itano93}.  Let $\rho_x$ be the parameter-dependent density
operator of the atoms on $\mc H_1$ and $\tau$ be the initial state of
an optical probe on $\mc H_0$.  Then the state after the optical
probing can be expressed as $U(\rho_x \otimes \tau)U^\dagger$, where
$U$ is a unitary operator on $\mc H_1\otimes\mc H_0$ that models the
atom-light interaction.  If an optical measurement, modeled by a
projection-valued measure $\Pi_0$ on $\mc H_0$, is performed and the
estimator in terms of the outcome $y$ is $b(y)$, then the $B$
observable in Lemma~\ref{lem_sep} and Corollary~\ref{cor_vn2} can be
expressed as
\begin{align}
B &= \int b(y) \Pi'(dy),
&
\Pi'(\cdot) &= U^\dagger \Bk{I_1\otimes \Pi_0(\cdot)} U,
\end{align}
and $\MSE_x$ is the error of this indirect measurement of the atoms.
The GCE $\mc G_*B$, on the other hand, is an atomic observable on
$\mc H_1$, and $\MSE_x'$ is the error of the direct atomic measurement
of $\mc G_*B$.  Then
$D_{\rho_x,\mc G}(B,\gce{\mc G}B) =\MSE_x - \MSE_x'$ can be regarded
as the regret due to the indirectness of the optical measurement. The
Bayesian setting can be studied similarly.

\subsection{\label{sec_sym}A sufficient channel for symmetric states}
Let $\{U_z: z \in \mc Z\}$ be a set of unitary operators on $\mc H_2$,
and suppose that $\rho_x$ is invariant to all of them, viz.,
\begin{align}
U_z \rho_x U_z^\dagger &= \rho_x \quad
\forall z \in \mc Z.
\label{rho_invar}
\end{align}
Examples include the symmetric states that are invariant to any
permutation of a tensor-powered Hilbert space---to be discussed
later---and optical states with random phases that are invariant to
any phase modulation.  $\rho_x$ is also invariant to the random
unitary channel
\begin{align}
\mc G \rho_x &= \int \mu(dz) U_z \rho_x U_z^\dagger = \rho_x
\label{ru}
\end{align}
for any probability measure $\mu$ on $(\mc Z,\Sigma_{\mc Z})$.
$\mc G$ is then a sufficient channel in Petz's sense, since another
equivalent condition for Petz's sufficiency is the existence of an
$x$-independent CPTP map that recovers $\rho_x$ from $\mc G \rho_x$
\cite{petz}. It is straightforward to compute the GCEs.
\begin{lemma}
\label{lem_invar}
Given Eqs.~(\ref{rho_invar}) and (\ref{ru}), a solution to any GCE is
\begin{align}
\gce{\mc G} B &= \int \mu(dz) U_z B U_z^\dagger,
\label{gce_invar}
\end{align}
which does not depend on $x$ if $\{U_z\}$ and $\mu$ do not.
\end{lemma}
See Appendix~\ref{app_proofs} for the proof.

\begin{corollary}
\label{cor_invar}
Given a family of states that are invariant to a set of unitaries
$\{U_z\}$, any estimator $B \in L_2(\rho_x)$, and the resulting local
error $\MSE_x$, there exists an averaged estimator given by
Eq.~(\ref{gce_invar}) that performs at least as well as $B$ for all
$x \in \mc X$.
\end{corollary}
\begin{proof}
Use Lemma~\ref{lem_invar} and Theorem~\ref{thm_qrbt}.
\end{proof}

If $\mc Z$ is a group and $\{U_z\}$ is a projective unitary
representation of the group that satisfies
$U_{z'}U_z = \omega(z',z) U_{z' z}$ for a complex scalar $\omega$ with
$|\omega| = 1$ \cite{holevo11}, then the left Haar measure $\tilde\mu$
on the group \cite{parth05} plays a special role, as the GCE with
respect to it, written as
\begin{align}
\gce{\tilde{\mc G}} B &\equiv \int \tilde\mu(dz) U_z B U_z^\dagger,
\label{gce_ultimate}
\end{align}
is invariant to any subsequent GCE for any random unitary channel, in
the sense that
\begin{align}
\gce{\mc G} \gce{\tilde{\mc G}} B &= \gce{\tilde{\mc G}} B
\end{align}
for any $\mu$. The left Haar measure is thus the ultimate choice that
gives the highest error reduction in the context of
Corollary~\ref{cor_invar}.

For a concrete example, let $\mc H_2 = \mc H_1^{\otimes n}$,
$\pi \in S_n$ be a permutation function of
$(1,\dots,n)$, and $S_n$ be the permutation group.
Define each unitary by
\cite{watrous}
\begin{align}
U_\pi(\ket{\psi_1}\otimes \dots \otimes \ket{\psi_n}) &= 
\ket*{\psi_{\pi^{-1}1}} \otimes \dots\otimes \ket*{\psi_{\pi^{-1}n}}
\label{U_perm}
\end{align}
for any $\{\ket{\psi_j} \in \mc H_1:j = 1,\dots,n\}$.  An operator
invariant to all the permutation unitaries is called
symmetric. Physically, a symmetric density operator corresponds to $n$
indistinguishable systems. A common example is
$\rho_x = \sigma_x^{\otimes n}$, where $\sigma_x$ is a density
operator on $\mc H_1$. The Haar measure is simply
$\tilde\mu(\pi) = 1/n!$, and the corresponding GCE is
\begin{align}
\gce{\tilde{\mc G}} B &= \frac{1}{n!}\sum_\pi U_\pi B U_\pi^\dagger,
\label{haar_perm}
\end{align}
which is a symmetrization. Furthermore, if one assumes
\begin{align}
B &= C\otimes I_1^{\otimes(n-m)},
&
C &\in \mc O(\mc H_1^{\otimes m}),
\label{B_perm}
\end{align}
then Eq.~(\ref{haar_perm}) leads to the quantum U-statistics
introduced by \guta{} and Butucea \cite{guta10}, as shown in
Appendix~\ref{app_U}.  The U-statistic is an unbiased estimator of
$a(x) = \trace \rho_x B = \trace \rho_x \gce{\mc G}B$.  The
simplest example is when $m = 1$ and
\begin{align}
\gce{\tilde{\mc G}} B &= \frac{1}{n}\sum_{l=1}^{n} 
I_1^{\otimes (l-1)}\otimes C \otimes I_1^{\otimes(n-l)},
\end{align}
which lowers the variance of $B$ by a factor of $n$ if
$\rho_x = \sigma_x^{\otimes n}$.

The derivation of the classical U-statistics by Rao-Blackwellization
is well known \cite{vaart}, and Corollary~\ref{cor_invar} is indeed
the appropriate quantum generalization.

\subsection{\label{sec_sum}A sufficient channel for direct-sum states}
Suppose now that $\BK{\rho_x:x \in \mc X}$ is a family of density
operators on a direct sum of Hilbert spaces given by
\begin{align}
\mc H = \bigoplus_{n \in \mc N} \mc H_n,
\label{H_sum}
\end{align}
and each $\rho_x$ is given by the direct sum
\begin{align}
\rho_x &= \bigoplus_{n \in \mc N} \sigma_x^{(n)},
\label{rho_sum}
\end{align}
where each $\sigma_x^{(n)}$ is a positive-semidefinite operator on
$\mc H_n$. A prominent example in optics is the multimode thermal
state, which will be discussed in Sec.~\ref{sec_therm}.  Let
$\Pi_n:\mc H\to \mc H_n$ be the projection operator onto $\mc H_n$.
Suppose that the Hilbert-space decomposition given by
Eq.~(\ref{H_sum}) is parameter-independent, such that all
$\{\Pi_n:n \in \mc N\}$ do not depend on $x$. Then the channel
\begin{align}
\mc G \rho_x &= \bigoplus_n \Pi_n \rho_x \Pi_n = \rho_x
\label{G_project}
\end{align}
is sufficient in Petz's sense. To compute the
GCEs with respect to Eqs.~(\ref{rho_sum}) and (\ref{G_project}), I
impose two more properties on the $\mc E$ map given by
\begin{align}
\mc E_{\sigma^{(1)}\oplus\sigma^{(2)}}(A_1 \oplus A_2) &= 
\bk{\mc E_{\sigma^{(1)}} A_1} \oplus \bk{\mc E_{\sigma^{(2)}} A_2},
\label{E_sum}
\\
\Pi_1 \bk{\mc E_{\sigma^{(1)}\oplus\sigma^{(2)}} A}\Pi_1 &= \mc E_{\sigma^{(1)}}(\Pi_1 A\Pi_1)
\label{E_project}
\end{align}
for any $A_j \in \mc O(\mc H_j)$, any
$A \in \mc O(\mc H_1\oplus \mc H_2)$, and any density operator on
$\mc H_1\oplus \mc H_2$ in the form of
$\sigma^{(1)}\oplus\sigma^{(2)}$. These properties are satisfied by
the products given by Eqs.~(\ref{jordan})--(\ref{root}) at least.
Then the GCE has the following solution.
\begin{lemma}
\label{lem_sum}
Given Eqs.~(\ref{rho_sum}) and (\ref{G_project}) and assuming a GCE in
terms of an $\mc E$ map that satisfies Eqs.~(\ref{E_sum}) and
(\ref{E_project}), a solution to the GCE is
\begin{align}
\gce{\mc G} B &= \bigoplus_n \Pi_n B \Pi_n,
\label{gce_sum}
\end{align}
which does not depend on $x$ if the projectors $\{\Pi_n\}$ do not.
\end{lemma}
See Appendix~\ref{app_proofs} for the proof.

The quantum Rao-Blackwell theorem can now be applied to
Eqs.~(\ref{rho_sum}) and (\ref{G_project}) to prove the following.
\begin{corollary}
\label{cor_sum}
Assume that the Hilbert space can be decomposed as Eq.~(\ref{H_sum})
and the projectors $\{\Pi_n:\mc H \to \mc H_n\}$ do not depend on the
unknown parameter $x$.  Given a density-operator family in the form of
a direct sum as per Eq.~(\ref{rho_sum}), any estimator
$B \in L_2(\rho_x)$, and the resulting local error $\MSE_x$, there
exists an estimator $\gce{\mc G} B$ given by Eq.~(\ref{gce_sum}),
also in the form of a direct sum, that performs at least as well as
$B$ for all $x \in \mc X$.
\end{corollary}
\begin{proof}
Use Lemma~\ref{lem_sum} and Theorem~\ref{thm_qrbt}.
\end{proof}
An example in optics is now in order.

\subsection{\label{sec_therm}Thermal-light sensing}
A multimode thermal optical state can be expressed as \cite{mandel}
\begin{align}
\rho_x &= \int d^{2J}\alpha \Phi_x(\alpha) \ket{\alpha}\bra{\alpha},
\\
\alpha &\equiv \mqty(\alpha_1 \\ \vdots \\ \alpha_J) \in \mathbb C^J,
\quad
\hat a_j\ket{\alpha} = \alpha_j\ket{\alpha},
\\
\Phi_x(\alpha) &= \frac{1}{\det(\pi\Gamma_x)}
\exp(-\alpha^\dagger \Gamma_x^{-1}\alpha),
\end{align}
where $\ket{\alpha}$ is a coherent state, $\hat a_j$ is the
annihilation operator for the $j$th mode,
$d^{2J}\alpha \equiv \prod_{j=1}^J d(\Re\alpha_j) d(\Im\alpha_j)$, and
$\Gamma_x$ is the positive-definite mutual coherence matrix. In
thermal-light sensing and imaging problems
\cite{stellar,nair_tsang,tnl,tsang19a,nair_tsang16,lupo}, $\Gamma_x$
is assumed to depend on the unknown parameter $x$.

Let $\mc H_n$ be the $n$-photon Hilbert space. Define a pure Fock
state with photon numbers $\bs m = (m_1,\dots,m_J) \in \mathbb N_0^J$
as
\begin{align}
\ket{\bs m} &\equiv 
\Bk{\prod_{j} \frac{(\hat a_j^\dagger)^{m_j}}{\sqrt{m_j!}}}\ket{\bs 0},
\end{align}
where $\ket{\bs 0}$ denotes the vacuum state. Let
$\norm{\bs m} \equiv \sum_j m_j$ be the total photon number. Then
$\{\ket{\bs m}: \norm{\bs m} = n\}$ is an orthonormal basis of
$\mc H_n$. In terms of the Fock basis, each matrix element of $\rho_x$
is given by
\begin{align}
\bra{\bs m}\rho_x\ket{\bs l} &= 
\int d^{2J}\alpha \Phi_x(\alpha) 
\prod_j e^{-|\alpha_j|^2}
\frac{\alpha_j^{m_j} (\alpha_j^*)^{l_j}}{\sqrt{m_j!l_j!}}.
\end{align}
The Gaussian moment theorem (see Eq.~(1.6-33) in Ref.~\cite{mandel})
implies that
\begin{align}
\bra{\bs m}\rho_x\ket{\bs l} &= 0
\textrm{ if } \norm{\bs m} \neq \norm{\bs l},
\end{align}
meaning that $\rho_x$ can be decomposed in the direct-sum form as
\begin{align}
\rho &= \bigoplus_{n=0}^\infty \sigma_x^{(n)},
\label{therm_sum}
\\
\sigma_x^{(n)} &= \sum_{\bs m,\bs l: \norm{\bs m} = \norm{\bs l} = n}
\bra{\bs m}\rho_x\ket{\bs l} \ket{\bs m}\bra{\bs l},
\end{align}
where each $\sigma_x^{(n)}$ is an operator on $\mc H_n$. Then
$\trace\sigma_x^{(n)}$ is the probability of having $n$ photons in
total and $\sigma_x^{(n)}/\trace\sigma_x^{(n)}$ is the conditional
$n$-photon state. The projectors can be written as
\begin{align}
\Pi_n &= \sum_{\bs m: \norm{\bs m} = n}\ket{\bs m}\bra{\bs m}.
\end{align}
Ignoring the mathematical complications due to the
infinite-dimensional Hilbert space, Corollary~\ref{cor_sum} can now be
applied to Eq.~(\ref{therm_sum}).

If an estimator is constructed from a photon-counting measurement with
respect to any set of optical modes, it can be expressed in a Fock
basis, which commutes with all the projectors $\{\Pi_n\}$. It follows
that the estimator is already in the direct-sum form given by
Eq.~(\ref{gce_sum}) and Corollary~\ref{cor_sum} offers no
improvement. On the other hand, notice that Eq.~(\ref{gce_sum}) must
commute with each projector $\Pi_n$, viz.,
\begin{align}
\Bk{\gce{\mc G} B, \Pi_n} &= 0 \quad \forall n \in \mc N.
\end{align}
If an estimator does not commute with all $\{\Pi_n\}$, such as one
obtained from homodyne detection, then the estimator does not have the
direct-sum form and has the potential to be improved by the quantum
Rao-Blackwellization. 

To introduce a more specific example, diagonalize $\Gamma_x$ in terms
of a diagonal matrix $D_x$ and a unitary matrix $V_x$ as
\begin{align}
\Gamma_x &= V_x D_x V_x^\dagger,
&
D_{jk,x} &= \lambda_{j,x}\delta_{jk},
\end{align}
where $\delta_{jk}$ is the Kronecker delta and each $\lambda_{j,x}$ is
an eigenvalue of $\Gamma_x$. I call $\{\lambda_{j,x}:j = 1,\dots,J\}$
the spectrum of the thermal state. With the change of variable
\begin{align}
\beta = V_x^\dagger\alpha,
\end{align}
$\Phi_x(\alpha) = \Phi_x(V_x\beta)$ becomes separable in terms of $\beta$.
Define also a unitary operator $\hat U_x$ by
\begin{align}
\hat U_x^\dagger \hat a_j \hat U_x = \sum_k V_{jk,x} \hat a_k \equiv \hat g_{j,x},
\label{eigenmodes}
\end{align}
such that $\ket{\alpha} = \ket{V_x\beta} = \hat U_x \ket{\beta}$. $\rho_x$
can then be expressed as
\begin{align}
\rho_x &= \sum_{\bs m} p_x(\bs m)\ket{\bs m,g}\bra{\bs m,g},
\label{rho_sep}
\end{align}
where 
\begin{align}
p_x(\bs m) &= \prod_j \frac{1}{1+\lambda_{j,x}}
\bk{\frac{\lambda_{j,x}}{1+\lambda_{j,x}}}^{m_j}
\end{align}
is separable into a product of Bose-Einstein distributions and
\begin{align}
\ket{\bs m,g} &\equiv \hat U_x\ket{\bs m} =
\Bk{\prod_j \frac{(\hat g_{j,x}^\dagger)^{m_j}}{\sqrt{m_j!}}}\ket{\bs 0}
\end{align}
is a Fock state with respect to the optical modes defined by
Eq.~(\ref{eigenmodes}). I call these optical modes the eigenmodes of
the thermal state. Now suppose that only the spectrum
$\{\lambda_{j,x}\}$ depends on the unknown parameter $x$, while $V$,
$U$, and thus $\{\hat g_j\}$ do not, meaning that the eigenmodes are
fixed. This assumption applies to the thermometry problem studied in
Ref.~\cite{nair_tsang} but does not apply to the
stellar-interferometry problem studied in Ref.~\cite{stellar} or the
subdiffraction-imaging problem studied in
Refs.~\cite{tnl,tsang19a,nair_tsang16,lupo}, because the eigenmodes
in the latter two cases vary with $x$.  With fixed eigenmodes, I can
define a more fine-grained $x$-independent projector as
\begin{align}
\Pi_{\bs m} &= \ket{\bs m,g}\bra{\bs m,g},
\label{project_U}
\end{align}
and apply Corollary~\ref{cor_sum} to Eq.~(\ref{rho_sep}). Plugging
Eq.~(\ref{project_U}) into Eq.~(\ref{gce_sum}) leads to the
Rao-Blackwell estimator
\begin{align}
\gce{\mc G} B &= \sum_{\bs m} \bra{\bs m,g}B\ket{\bs m,g}
\ket{\bs m,g}\bra{\bs m,g},
\label{gce_fock}
\end{align}
which can be implemented by counting the photons in the eigenmodes and
using $\bra{\bs m,g}B\ket{\bs m,g}$ as the estimator.
Equation~(\ref{gce_fock}) then implies the following corollary.
\begin{corollary}
\label{cor_therm}
Suppose that a real scalar parameter $a(x)$ of the spectrum
$\{\lambda_{j,x}\}$ of a thermal-state family is to be estimated and
the eigenmodes are parameter-independent.  Given any measurement, any
estimator, and the resulting local error $\MSE_x$, there exists an
estimator with eigenmode photon counting that performs at least as
well for all $x \in \mc X$.
\end{corollary}
\begin{proof}
  Corollary~\ref{cor_vn2} means that only von Neumann measurements
  need to be considered.  Any estimator with any von Neumann
  measurement can be Rao-Blackwellized to become Eq.~(\ref{gce_fock}),
  which can be implemented by eigenmode photon
  counting. Corollary~\ref{cor_sum} then guarantees that the $\MSE_x'$
  achieved by Eq.~(\ref{gce_fock}) can do at least as well for all
  $x \in \mc X$.
\end{proof}
In this example, the family of density operators given by
Eq.~(\ref{rho_sep}) and the $\gce{\mc G}B$ given by
Eq.~(\ref{gce_fock}) happen to commute with one another, but the
original estimator $B$ need not commute with the others, unlike
Sinha's assumption in Ref.~\cite{sinha22}; see
Appendix~\ref{app_sinha} for a brief discussion of his theory. Take
homodyne detection for example. An estimator constructed from homodyne
detection can be framed as
\begin{align}
B &= b(\hat q),
\label{B_hom}
\end{align}
where $\hat q$ is a vectoral quadrature operator that is a linear
function of $\{\hat a_j\}$.  Equation~(\ref{B_hom}) does not commute
with Eq.~(\ref{project_U}) in general, but Corollary~\ref{cor_therm}
still applies to it.

To demonstrate the possible improvement through an even more specific
example, suppose that the spectrum is flat and
$a(x) = \lambda_{j,x} = x$, the mean photon number per mode, is the
parameter of interest.  With homodyne detection, Eq.~(\ref{B_hom}) is
an unbiased estimator of $x$ if
\begin{align}
B &= \frac{1}{J}\sum_j \hat q_j^2 - \frac{1}{2},
&
\hat q_j& = \frac{\hat g_j+\hat g_j^\dagger}{\sqrt{2}}.
\end{align}
The Rao-Blackwell estimator given by Eq.~(\ref{gce_fock}),
on the other hand, can be expressed as
\begin{align}
\gce{\mc G}B &= \frac{1}{J}\sum_j \hat g_j^\dagger \hat g_j.
\end{align}
With the thermal state, it is straightforward to show that
\begin{align}
\MSE_x &= \frac{2}{J} \bk{x + \frac{1}{2}}^2,
&
\MSE_x' &= \frac{1}{J} \bk{x^2+x},
\label{MSE_therm}
\end{align}
which are plotted in Fig.~\ref{rb_therm}, demonstrating
the domination of the Rao-Blackwell estimator.

\begin{figure}[htbp!]
\centerline{\includegraphics[width=0.48\textwidth]{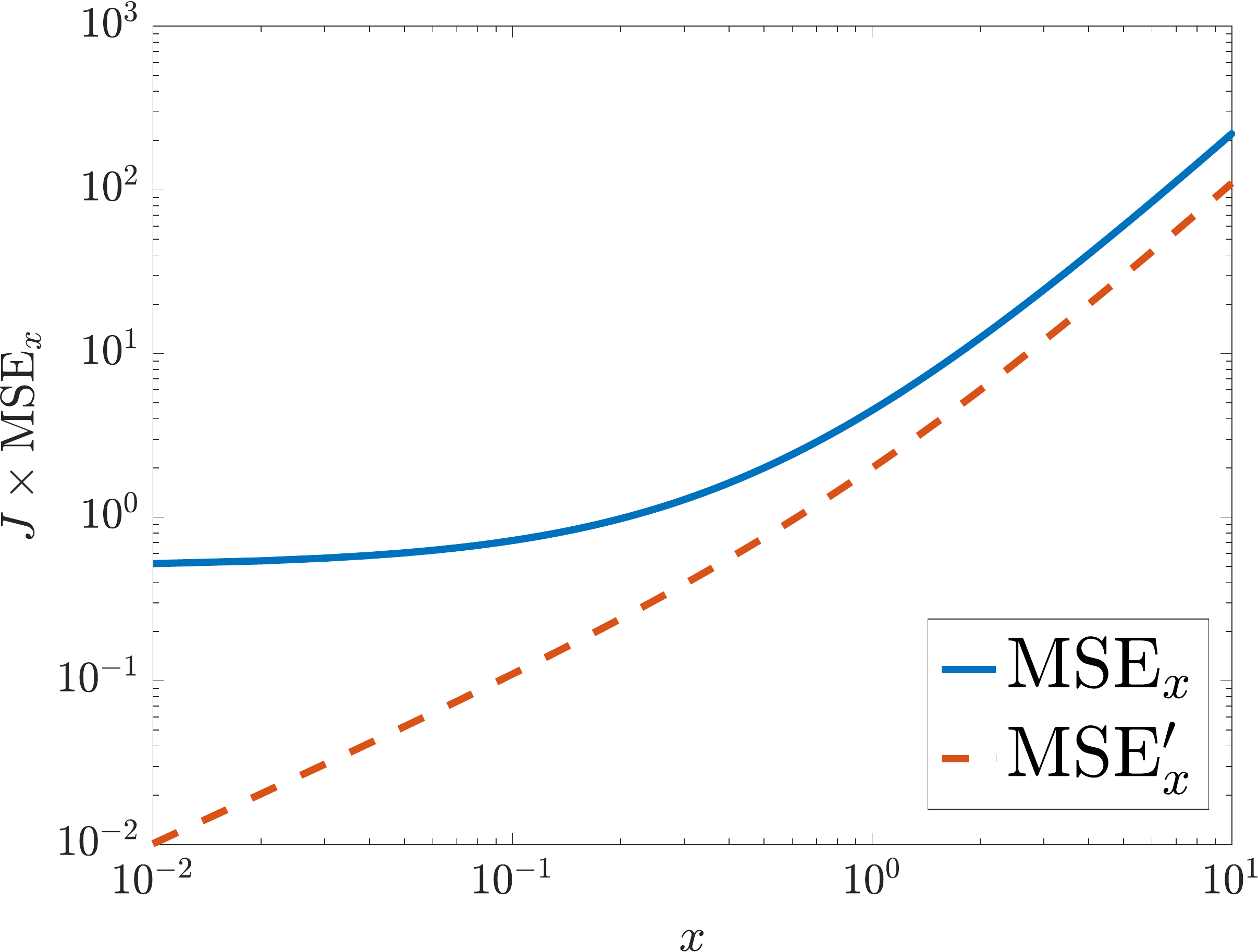}}
\caption{\label{rb_therm}Comparison of the mean-square error obtained
  by homodyne detection ($\MSE_x$) and that by photon counting
  ($\MSE_{x}'$) in estimating the mean photon number per mode $x$ of a
  thermal state. The plot is in log-log scale, both axes are
  dimensionless, and the errors are normalized with respect to $J$,
  the number of optical modes.  The improvement can be regarded as a
  result of the quantum Rao-Blackwellization. }
\end{figure}

Corollary~\ref{cor_therm} is reminiscent of the optimality of photon
counting for thermometry proved in Ref.~\cite{nair_tsang} in terms of
a quantum Cram\'er-Rao bound. Corollary~\ref{cor_therm} is more
general because it applies directly to the local mean-square error of
any biased or unbiased estimator and allows the parametrization of the
spectrum $\{\lambda_{j,x}\}$ and the parameter of interest $a(x)$ to
be general. The superiority of photon counting over homodyne detection
for random displacement models has also been noted in many other
contexts \cite{ng16,gorecki22,shi23,tsang23}, although those works,
like Ref.~\cite{nair_tsang}, rely on the quantum Cram\'er-Rao
  bound as well.

If the eigenmodes vary with $x$, as in the problems of stellar
interferometry \cite{stellar} and subdiffraction imaging
\cite{tnl,tsang19a,nair_tsang16,lupo}, then Eq.~(\ref{gce_fock}) may
not be a valid estimator, because $x$ is unknown and the measurement
may not be realizable. It is an interesting open question whether the
quantum Rao-Blackwell theorem can offer any insight about those
problems as well, beyond the optimality of the direct-sum form in
Corollary~\ref{cor_sum}. I speculate on two potential directions of
future research:
\begin{enumerate}
\item Even if the $\gce{\mc G}$ map may not be constant in general,
  $\gce{\mc G}B$ for a particular $B$ may happen to be constant and
  still a valid estimator.

\item Even if $\gce{\mc G}B$ varies with $x$, Eq.~(\ref{qrbt}) remains
  valid and can be used as a lower bound on $\MSE_x$, in which case
  $\MSE_x \ge \MSE_x'$ is an oracle inequality. An estimator that
  approximates $\gce{\mc G}B$, via an adaptive protocol for example
  \cite{wiseman_milburn}, may still enjoy an error close to $\MSE_x'$.
\end{enumerate}

\section{\label{sec_con}Conclusion}
This work cements the Jordan-product GCE and the associated divergence
as essential concepts in quantum metrology. In the Bayesian setting,
the GCE is found to relate the optimal estimators for a sequence of
channels. In the frequentist setting, the GCE is found to give a
quantum Rao-Blackwell theorem, which can improve a quantum estimator
in the same manner as the classical version does and reveal the
optimal forms of the estimators in common scenarios. In both settings,
the divergence is found to play a significant role in determining the
gap between the estimation errors before and after a channel is
applied. Given these operational meanings, even the purists
\cite{gough20,james21} can no longer dismiss the GCE and the
divergence as pointless concepts. For the more open minds, the
concepts have unveiled a new suite of methods for the study of
decoherence and the design of better measurements in quantum
metrology.

Many open problems remain. First, it should be possible to generalize
the theory here rigorously for infinite-dimensional Hilbert
spaces. Second, it may be possible to generalize the quantum
Rao-Blackwell theorem here for other convex loss functions beyond the
square loss, in the same manner as the classical version
\cite{lehmann98} or Sinha's versions \cite{sinha22}.  Third, there may
be a deeper relation between Petz's sufficiency and the constant GCE
condition desired here, beyond the specific examples in this
work. Fourth, there should be no shortage of further interesting
examples and applications of the theory here for quantum
metrology. Last but not the least, the strategy of using quantum
metrology to give operational meanings to GCEs may be generalizable
for other versions of GCEs and other metrological tasks, such as
multiparameter estimation, thus expanding the fundamental role of GCEs
in both quantum metrology and quantum probability theory.

\section*{Acknowledgment}
I acknowledge helpful discussions with Arthur Parzygnat.  This
research is supported by the National Research Foundation, Singapore,
under its Quantum Engineering Programme (QEP-P7).

\appendix

\section{\label{app_predict}Quantum prediction}
In analogy with Eq.~(\ref{gce_def}), $\mc F^*$, the Hilbert-Schmidt
adjoint of the CPTP map $\mc F$ that models a channel, obeys the
following interesting formula for any $\mc E$:
\begin{align}
\mc F^* B &= \argmin_{A \in L_2(\sigma)} D_{\sigma,\mc F}(A,B).
\label{predict}
\end{align}
In other words, $\mc F^*$ is the optimal prediction in the same way
$\gce{\mc F}$ is the optimal retrodiction.  The notations of $\mc F^*$
and $\gce{\mc F}$ also coincide with those of the pullback and the
pushforward in differential geometry \cite{lee03}, respectively, and
indeed $\mc F^*$ and $\gce{\mc F}$ behave like those operations.

With the adjoint relation between $\mc F^*$ and $\gce{\mc F}$ given by
Eq.~(\ref{gce}) in terms of the weighted inner product, $\mc F^*$ can
be shown to obey
\begin{align}
\mc E_{\sigma}\mc F^*B &= (\mc F_*)^* \mc E_{\mc F\sigma} B,
\label{predict_exp}
\end{align}
in analogy with Eq.~(\ref{gce_exp}), while the minimum divergence in
Eq.~(\ref{predict}) can be expressed as
\begin{align}
\mc D_{\sigma,\mc F}(\mc F^*B,B) &= 
\min_{A \in L_2(\sigma)} D_{\sigma,\mc F}(A,B)
\nonumber\\
&= \norm{B}_{\mc F\sigma}^2 -\norm{\mc F^*B}_{\sigma}^2,
\end{align}
in analogy with Eq.~(\ref{Dmin}).  The divergence given by
Eq.~(\ref{D}) can also be rewritten in a time-symmetric form as
\begin{align}
D_{\sigma,\mc F}(A,B) &= 
\Avg{A-\mc F^*B,A}_\sigma + \Avg{B-\gce{\mc F} A,B}_{\mc F\sigma}.
\end{align}
Together with Eqs.~(\ref{gce_def})--(\ref{Dmin}), these formulas
complete a satisfying time-symmetric theory of quantum inference.

\section{\label{app_ce}Classical conditional expectation}
The classical concepts in this Appendix are all special cases of the
quantum concepts in Sec.~\ref{sec_review} and
Appendix~\ref{app_predict}; see Table~\ref{tab_corr} for the
correspondence.

\begin{table*}[htbp!]
\begin{tabularx}{0.95\textwidth}{|l|X|X|}
  \hline
  Concept & Quantum & Classical \\
  \hline
  Prior state & Density operator $\sigma$ & 
Probability distribution $P_X(x)$ \\
  Channel & CPTP map $\mc F$ & $\sum_x P_{Y|X}(y|x) (\cdot)$ \\
  Unweighted inner product & Hilbert-Schmidt
  $\avg{A,B}_{\HS} \equiv \trace A^\dagger B$ &
  $\avg{a,b} \equiv \sum_\lambda [a(\lambda)]^* b(\lambda)$\\
  Weighted inner product &
  $\avg{A,B}_\rho \equiv \trace A^\dagger \mc E_\rho B$
  & $\avg{a,b}_{P}  \equiv \sum_\lambda [a(\lambda)]^* b(\lambda) P(\lambda)$\\
  Weighted norm & $\norm{A}_\rho \equiv \sqrt{\avg{A,A}_\rho}$ &
  $\norm{a}_P \equiv \sqrt{\avg{a,a}_P}$ \\
  Divergence between variables & $D_{\sigma,\mc F}(A,B)$
  (Definition~\ref{def_div}) &
  $\sum_{x,y}|a(x)-b(y)|^2 P_{XY}(x,y)$\\
  Optimal prediction & $\mc F^*$ (Hilbert-Schmidt adjoint of $\mc F$) &  
$\sum_y (\cdot) P_{Y|X}(y|x)$ \\
  Optimal retrodiction & $\mc F_*$ (Definition~\ref{def_gce}) &
  $\sum_x (\cdot) P_{X|Y}(x|y)$\\
\hline
\end{tabularx}

\caption{\label{tab_corr}Concepts in quantum inference and their
  classical versions.}
\end{table*}

Let $X$ and $Y$ be classical random variables with countable sample
spaces $\mc X$ and $\mc Y$, respectively.  Suppose that $X$ is
generated at an earlier time than $Y$. Let $P_X(x)$ be the probability
distribution of $X$ and $P_{Y|X}(y|x)$ be the probability distribution
of $Y$ conditioned on $X = x$. The unconditional distribution of $Y$
is given by
\begin{align}
  P_Y(y) &= \sum_x P_{Y|X}(y|x) P_X(x),
\label{markov_map}
\end{align}
which can be expressed as $P_Y = \mc F P_X$. The divergence between
two complex random variables $a(X)$ and $b(Y)$ can be defined as
\begin{align}
D_{P_X,\mc F}(a,b) &\equiv \sum_{x,y} \abs{a(x)-b(y)}^2 P_{Y|X}(y|x) P_X(x)
\label{Dc}\\
&= \norm{a}^2_{P_X} -  2 \Re \Avg{\mc F^* b,a}_{P_X} + \norm{b}^2_{P_Y},
\label{Dc2}
\end{align}
where the inner product and the norm are defined in
Table~\ref{tab_corr} and $\mc F^*$ is the adjoint of $\mc F$ with
respect to the unweighted inner product in Table~\ref{tab_corr}, such
that
\begin{align}
(\mc F^* b)(x) &= \sum_y b(y) P_{Y|X}(y|x).
\label{Fac}
\end{align}
In other words, $\mc F^*b$ is the conditional expectation of $b(Y)$
given $X$.

Let $L_2(P)$ be the Hilbert space of random variables in terms of the
inner product $\avg{\cdot,\cdot}_P$ defined in Table~\ref{tab_corr}.
The conditional expectation $\mc F^*b$ is the optimal predictor of
$b(Y)$ as a function of $X$, viz.,
\begin{align}
\mc F^* b &= \argmin_{a \in L_2(P_X)} D_{P_X,\mc F}(a,b).
\label{predict_c}
\end{align}
Similarly, the conditional expectation of $a(X)$ given $Y$ is the
optimal retrodictor of $a(X)$ as a function of $Y$, viz.,
\begin{align}
  \mc F_* a &= \argmin_{b \in L_2(P_Y)} D_{P_X,\mc F}(a,b),
  \\
  (\mc F_* a)(y) &= \sum_x a(x) P_{X|Y}(x|y),
  \\
  P_{X|Y}(x|y) &= \frac{P_{Y|X}(y|x) P_X(x)}{P_Y(y)}.
\label{bayes}
\end{align}
Equation~(\ref{bayes}) is, of course, the Bayes theorem. It is
straightforward to show that the formulas presented thus far are
special cases of the quantum formulas in Sec.~\ref{sec_review} and
Appendix~\ref{app_predict}; this is done by assuming the diagonal
forms
\begin{align}
\sigma &= \sum_x P_X(x)\ket{x}\bra{x},
\\
A &= \sum_x a(x)\ket{x}\bra{x},
\\
\mc F\sigma &= \sum_{y,x}P_{Y|X}(y|x) \bra{x}\sigma \ket{x}\ket{y} \bra{y},
\\
B &= \sum_y b(y)\ket{y}\bra{y},
\end{align}
where $\{\ket{x}:x \in \mc X\}$ is an orthonormal basis of $\mc H_1$
and $\{\ket{y}: y \in \mc Y\}$ is an orthonormal basis of $\mc H_2$.

Since the joint distribution $P_{XY}(x,y) = P_{Y|X}(y|x) P_X(x)$ is
another probability distribution, Eq.~(\ref{Dc}) can be written as the
squared norm
\begin{align}
D_{P_X,\mc F}(a,b) &= \norm{a-b}^2_{P_{XY}},
\end{align}
and since $L_2(P_X)$ and $L_2(P_Y)$ are subspaces of $L_2(P_{XY})$,
the Hilbert projection theorem implies that the conditional
expectation is a projection, viz.,
\begin{align}
\mc F^* b &= \argmin_{a \in L_2(P_X)} \norm{a-b}_{P_{XY}}^2 = \Pi[b|L_2(P_X)],
\\
\mc F_* a &= \argmin_{b \in L_2(P_Y)} \norm{a-b}_{P_{XY}}^2 = \Pi[a|L_2(P_Y)],
\end{align}
where $\Pi[u|V]$ is the projection of a Hilbert-space element $u$ into
a subspace $V$. In the quantum case, $\mc F^*$ and $\mc F_*$ are also
projections into appropriate Hilbert spaces if the equality condition
in Eq.~(\ref{D_dist}) holds.

Because $X$ is assumed to be generated earlier than $Y$ and
experiments follow an arrow of time, it is more straightforward to
obtain $P_{Y|X}$ by experiments than $P_{X|Y}$. This is the reason why
one commonly assumes that $P_{Y|X}$ rather than $P_{X|Y}$ is known in
an inference problem, and the prediction formula in terms of $P_{Y|X}$
looks simpler than the retrodiction formula in terms of the Bayes
theorem. Probability theory in itself is agnostic to the arrow of
time, and if $P_{X|Y}$ is given instead, the retrodiction formula
would look simpler and the prediction formula would require the Bayes
theorem. The same arrow of time is commonly assumed in quantum theory,
where the $\mc F$ map models the evolution of a quantum system forward
in time. It is therefore unsurprising that the prediction map
$\mc F^*$ in terms of $\mc F$ also looks simpler than the retrodiction
map $\mc F_*$ in terms of Eq.~(\ref{gce_exp}), which is a
generalization of the Bayes theorem. If the map $(\mc F_*)^*$ is given
instead, then the retrodiction map $\mc F_*$ is simply its
Hilbert-Schmidt adjoint, while the prediction map $\mc F^*$ is given
by the complicated Eq.~(\ref{predict_exp}).

\section{\label{app_gauss}GCE for Gaussian systems}
I first briefly review the theory of quantum Gaussian systems,
following Chap.~12 in Ref.~\cite{holevo19}. Let $\mc H_1$ be the
Hilbert space for $s$ bosonic modes.  On $\mc H_1$, define the
canonical observables as
\begin{align}
Q &\equiv \mqty(q_1& p_1 & \dots & q_s & p_s)^\top,
&
[q_j,p_k] &= i \delta_{jk},
\end{align}
and the Weyl operator as
\begin{align}
W(z) &\equiv \exp(iQ^\top z),
\\
z &\equiv \mqty(x_1 & y_1 & \dots & x_s & y_s)^\top \in \mathbb R^{2s},
\end{align}
where $\top$ denotes the transpose. If $\sigma$ is a Gaussian state,
its characteristic function can be expressed as
\begin{align}
\phi(z) &\equiv \trace \sigma W(z)
= \exp\qty(im^\top z - \frac{1}{2} z^\top \Sigma z),
\label{phi_gauss}
\end{align}
where $m \in \mathbb R^{2s}$ is the mean vector and
$\Sigma \in \mathbb R^{2s\times 2s}$ is the covariance matrix of the
Gaussian state.  $\Sigma$ is symmetric, positive-semidefinite, and
must observe an uncertainty relation that need not
concern us here.

Similar to the preceding definitions, let $\mc H_2$ be the Hilbert
space for $t$ bosonic modes and define $\tilde Q$ and
$\tilde W(\zeta)$ as the canonical observables and the Weyl operator
on $\mc H_2$, respectively. If
$\mc F:\mc O(\mc H_1)\to \mc O(\mc H_2)$ is a CPTP map that models a
Gaussian channel, it can be defined by
\begin{align}
\mc F^* \tilde W(\zeta) &= f(\zeta) W(F^\top\zeta) ,
\label{F_gauss}
\\
f(\zeta) &= \exp\qty(i l^\top \zeta - \frac{1}{2} \zeta^\top R \zeta),
\label{F_gauss2}
\end{align}
where $F\in \mathbb R^{2t\times 2s}$ is a transition matrix,
$l \in \mathbb R^{2t}$ is the mean displacement introduced by the
channel, and $R \in \mathbb R^{2s\times 2s}$ is the channel covariance
matrix.  $F$ and $R$ must obey a certain matrix inequality for the map
to be CPTP, but again the inequality need not concern us here. With
the Gaussian input state and the Gaussian channel, the output state
remains Gaussian and its characteristic function is given by
\begin{align}
\tilde\phi(\zeta) &\equiv \trace (\mc F\sigma) \tilde W(\zeta)
= f(\zeta)\phi(F^\top\zeta) 
\label{phi_tilde}
\\
&= \exp\bk{i \tilde m^\top \zeta - \frac{1}{2} \zeta^\top \tilde\Sigma \zeta},
\label{phi_tilde2}
\\
\tilde m &= F m + l,
\label{m_tilde}
\\
\tilde\Sigma &= F \Sigma F^\top + R.
\label{Sigma_tilde}
\end{align}
An explicit GCE formula can now be presented.
\begin{proposition}
\label{prop_gauss}
Assume a Gaussian state defined by Eq.~(\ref{phi_gauss}), a
Gaussian channel defined by Eqs.~(\ref{F_gauss}) and (\ref{F_gauss2}),
a quadrature operator given by
\begin{align}
A &= u^\top Q, & u &= \mqty(u_1 & \dots & u_{2s})^\top \in \mathbb R^{2s},
\label{A_quad}
\end{align}
and the $\mc E$ map given by the Jordan product in Eq.~(\ref{jordan}).
A solution to the GCE is
\begin{align}
\gce{\mc F} A &= u^\top \Bk{m + K\bk{\tilde Q - Fm- l}},
\label{gce_gauss}
\\
K &\equiv \Sigma F^\top \bk{F\Sigma F^\top + R}^{-1},
\label{kalman_gain}
\end{align}
while the divergence is
\begin{align}
D_{\sigma,\mc F}(A,\gce{\mc F} A)
&= u^\top \bk{\Sigma - KF\Sigma} u.
\label{D_gauss}
\end{align}
\end{proposition}
See Appendix~\ref{app_proofs} for the proof.

It is interesting to note that Eqs.~(\ref{gce_gauss})--(\ref{D_gauss})
are identical to the formulas for the classical conditional
expectation $\expect(A|Y)$ and its mean-square error
$\expect\{ [\expect(A|Y) - A]^2\}$ when $A = u^\top X$, $Y = F X + Z$,
and $X \sim N(m,\Sigma)$ and $Z \sim N(l,R)$ are independent normal
random variables \cite{anderson}. Here, the canonical observables $Q$
and $\tilde Q$ play the roles of $X$ and $Y$, respectively. When
$\mc F$ is a measurement map, similar formulas have been derived in
Refs.~\cite{yanagisawa,smooth,smooth_pra1,tsang22b} and may be useful
for studying waveform estimation \cite{twc} beyond the stationary
assumption.

\section{\label{app_vn}von Neumann measurement}
Let $B$ be a self-adjoint operator on a Hilbert space $\mc H$.  Then
$B$ admits the spectral representation
\cite{reed_simon,holevo11}
\begin{align}
B &= \int b \Pi(db)
\label{pvm}
\end{align}
in terms of a projection-valued measure $\Pi$ on the Borel space
$(\mathbb R,\Sigma_{\mathbb R})$. A projection-valued measure, also
called an orthogonal resolution of identity, is a POVM on
$(\mc Y,\Sigma_{\mc Y})$ with the additional properties that $\Pi(C)$
for any $C \in \Sigma_{\mc Y}$ is a projection operator and
$\Pi(C_1)\Pi(C_2) = 0$ if $C_1 \cap C_2 = \varnothing$
\cite{holevo11}.  The projection-valued measure in Eq.~(\ref{pvm}) is
unique to $B$ in the sense that no other projection-valued measure can
satisfy Eq.~(\ref{pvm}) for a given $B$.  A von Neumann measurement of
$B$ is defined as a measurement with the projection-valued measure
$\Pi$ of $B$ as the POVM.

If $\mc H$ is finite-dimensional, $B$ has a finite number of
eigenvalues, and the spectral representation is much simpler and can
be written as
\begin{align}
B &= \sum_{b \in \Lambda_B} b \Pi(b),
\label{pvm_finite}
\end{align}
where $\Lambda_B \subset \mathbb R$ is the set of $B$'s eigenvalues
and $\{\Pi(b): b \in \Lambda_B\}$ are projection operators that obey
$\Pi(b) \Pi(c) = \delta_{bc}\Pi(b)$ and $\sum_b \Pi(b) = I$.  The set
of outcomes can then be restricted to $\Lambda_B$.

Consider now a converse situation, where the POVM is given by a
projection-valued measure $\Pi'$ on $(\mc Y,\Sigma_{\mc Y})$ and each
outcome $y \in \mc Y$ is processed by a function
$b:\mc Y \to \mathbb R$ to produce a final outcome $b(y)$.  The
measurement $\Pi'$ together with the data processing by $b(y)$ can be
interpreted as a von Neumann measurement of
\begin{align}
B &= \int b(y) \Pi'(dy).
\label{pvm_b}
\end{align}
This is because Eq.~(\ref{pvm_b}) can be expressed in the spectral
representation given by Eq.~(\ref{pvm}), with
\begin{align}
\Pi(C) &= \Pi'[b^{-1}(C)],
&
b^{-1}(C) &\equiv \BK{y: b(y) \in C},
\end{align}
and the probability $\trace \Pi'[b^{-1}(C)] \rho$ of each event
$C \in \Sigma_{\mathbb R}$ generated by $(\Pi',b)$ coincides with the
probability $\trace \Pi(C)\rho$ from a von Neumann measurement of
$B$. If $\mc Y$ is countable, then $B = \sum_y b(y) \Pi'(y)$ can be
expressed as Eq.~(\ref{pvm_finite}), with
\begin{align}
\Lambda_B &= \BK{b(y): y \in \mc Y},
&
\Pi(c) &= \sum_{y: b(y) = c} \Pi'(y).
\end{align}

\section{\label{app_dp}The dynamic-programming algorithm}
The dynamic-programming algorithm is based on the so-called principle
of optimality for Markov decision processes. For the process given by
Eqs.~(\ref{markov}) and (\ref{Dsum}), the principle states that, if
$(\tilde{\mc F}_1,\dots,\tilde{\mc F}_N)$ are the optimal ``controls''
that minimize the total ``cost'' given by Eq.~(\ref{Dsum}), then
$(\tilde{\mc F}_k,\dots,\tilde{\mc F}_N)$ must also be the optimal
controls that minimize the ``cost-to-go'' at time $k$, defined as
\begin{align}
J_k &\equiv \sum_{n=k}^N g(s_{n},\mc F_n).
\end{align}
Based on this principle, the algorithm starts by choosing the final
control $\mc F_N$ that minimizes the cost-to-go
$J_N = g(s_{N},\mc F_N)$. Let this minimum be
\begin{align}
\tilde J_{N}(s_{N}) &\equiv \min_{\mc F_N} J_N
= \min_{\mc F_N} g(s_{N},\mc F_N).
\end{align}
The optimal $\mc F_N$ obtained this way is a function
$\tilde{\mc F}_N(s_N)$ of the state $s_N$. Next, write
$\tilde J_{N}(s_{N}) = \tilde J_{N}[f(s_{N-1},\mc F_{N-1})]$ using
Eq.~(\ref{markov}) and find the control $\mc F_{N-1}$ that minimizes
the cost-to-go at time $k = N-1$, given by
\begin{align}
\tilde J_{N-1}(s_{N-1})
&\equiv 
\min_{(\mc F_{N-1},\mc F_N)} J_{N-1}
\\
&= \min_{\mc F_{N-1}} 
\Big\{g(s_{N-1},\mc F_{N-1}) + 
\nonumber\\&\quad 
\tilde J_{N}\Bk{f(s_{N-1},\mc F_{N-1})}\Big\}.
\end{align}
The optimal $\mc F_{N-1}$ is again a function
$\tilde{\mc F}_{N-1}(s_{N-1})$ of the state $s_{N-1}$.  This procedure
continues with
\begin{align}
\tilde J_{k}(s_{k}) &\equiv 
\min_{(\mc F_k,\dots,\mc F_N)} J_k
\\
&= \min_{\mc F_k} 
\left\{g(s_{k},\mc F_k) + 
%\right. \nonumber\\&\quad \left.
\tilde J_{k+1}\Bk{f(s_{k},\mc F_k)}\right\},
\\
\tilde{\mc F}_k(s_k) &= \argmin_{\mc F_k} \BK{\dots}
\end{align}
for $k = N-2,N-3,\dots$ until $k = 1$, when the complete optimal
control law $(\tilde F_1(s_1),\dots,\tilde F_N(s_N))$ has been found
and $\tilde J_1(s_1)$ is the minimum total cost.

\section{\label{app_rb}Classical Rao-Blackwell theorem}
To derive the classical Rao-Blackwell theorem from the quantum theorem
given by Theorem~\ref{thm_qrbt}, assume the diagonal forms
\begin{align}
\rho_x &= \sum_y P_{Y|X}(y|x) \ket{y}\bra{y},
\\
B &= \sum_y b(y) \ket{y}\bra{y},
\\
\mc G \rho_x &= \sum_{z,y} P_{Z|Y}(z|y) \bra{y}\rho_x\ket{y}\ket{z}\bra{z}
\\
&= \sum_{z} P_{Z|X}(z|x) \ket{z}\bra{z},
\\
P_{Z|X}(z|x) &= \sum_y P_{Z|Y}(z|y) P_{Y|X}(y|x),
\end{align}
where $\{\ket{y}:y \in \mc Y\}$ is an orthonormal basis of $\mc H_2$,
$\{\ket{z}: z \in \mc Z\}$ is an orthonormal basis of $\mc H_3$, $X$,
$Y$, and $Z$ are classical variables, and $P_{O|O'}$ is the
probability distribution of $O$ conditioned on a value of $O'$.  For
the estimation problem, $X$ is the hidden parameter fixed at $X = x$,
$Y$ is the observation, and $Z$ is a statistic generated from $Y$
without knowing $X$, such that $P_{Z|Y,X}(z|y,x) = P_{Z|Y}(z|y)$.  The
parameter of interest $a(x)$ is assumed to be real for simplicity. The
local error given by Eq.~(\ref{MSE}) becomes
\begin{align}
\MSE_x &= \sum_y \Bk{b(y)-a(x)}^2 P_{Y|X}(y|x),
\label{MSE_c}
\end{align}
which agrees with the classical definition.  Equation~(\ref{gce_exp}) 
for the GCE $\gce{\mc G}B$ becomes
\begin{align}
\mc E_{\mc G\rho_x}\gce{\mc G}B &= \mc G \mc E_{\rho_x}B,
\end{align}
a solution of which is
\begin{align}
\gce{\mc G} B &= \sum_z c(z,x) \ket{z}\bra{z},
\\
c(z,x) &= \sum_{y} P_{Y|Z,X}(y|z,x) b(y),
\label{rb_c}
\\
P_{Y|Z,X}(y|z,x) &= \frac{P_{Z|Y}(z|y) P_{Y|X}(y|x)}{P_{Z|X}(z|x)}.
\label{PYZX}
\end{align}
$c(z,x)$ is the expectation of $b(Y)$ conditioned on $Z = z$ and
$X = x$. If $c(z,x) = c(z)$ does not depend on $x$, then it is a valid
estimator as a function of the statistic $Z$, and its error as per
Eq.~(\ref{MSE2}) becomes
\begin{align}
\MSE_x' &= \sum_z \Bk{c(z) - a(x)}^2 P_{Z|X}(z|x).
\label{MSE2_c}
\end{align}
It follows from Theorem~\ref{thm_qrbt}
and Eq.~(\ref{Dmin}) that the difference between Eqs.~(\ref{MSE_c})
and (\ref{MSE2_c}) is
\begin{align}
&\quad \MSE_x - \MSE_x' 
\nonumber\\
&= \sum_y [b(y)]^2 P_{Y|X}(y|x) - \sum_z [c(z)]^2 P_{Z|X}(z|x),
\label{crbt}
\end{align}
which is nonnegative and coincides with a form of the Rao-Blackwell
theorem (see, for example, Problem~1.7.9 on p.~73 in
Ref.~\cite{lehmann98}).

$Z$ is called a sufficient statistic for the estimation problem if the
distribution $P_{Y|Z,X}(y|z,x)$ given by Eq.~(\ref{PYZX}) does not
depend on $x$ \cite{lehmann98}. Then Eq.~(\ref{rb_c}) also does not
depend on $x$ for any $b(y)$ and is always a valid estimator.

\section{\label{app_compare}Comparison of the
Bayesian and frequentist settings}
Both the monotonicity of the Bayesian error given by
Corollary~\ref{cor_mono} and the error reduction due to the quantum
Rao-Blackwell theorem in Theorem~\ref{thm_qrbt} are unsurprising
results given their classical origins, but they may be confusing
in that they seem to say opposite things about the effect of a
channel. I offer a clarification here.

First, note that the Bayesian setting concerns the ``global'' error
$D_{\sigma,\mc F}(A,B)$ only, whereas the frequentist setting concerns
the local error $\MSE_x$ as a function of the unknown parameter $x$.
The global error is a cruder measure because it is only an average of
the local error given by
\begin{align}
D_{\sigma,\mc F}(A,B) &= \sum_x P_X(x) \MSE_x,
\label{bayes_freq}
\end{align}
assuming Eqs.~(\ref{MSE_bayes}) and (\ref{MSE}).

Second, the Bayesian results in Sec.~\ref{sec_bayes}, and
Corollary~\ref{cor_mono} in particular, concern only the estimators
$\gce{\mc F} A$ and $\gce{\mc G} \gce{\mc F} A$ that are
optimal with respect to the global error. Theorem~\ref{thm_qrbt} in
the frequentist setting, on the other hand, is about the local errors
of an estimator $B$ and its Rao-Blackwellization $\gce{\mc G}B$,
with no special assumptions about the original estimator $B$. The
theorem also says nothing about whether the Rao-Blackwell estimator is
optimal in the global sense, only that it is at least as good as the
original.

Third, note that the Personick estimators considered in the Bayesian
setting do not depend on the unknown parameter and are naturally
realizable, and Corollary~\ref{cor_mono} applies to any channel. In
the frequentist setting, the channel and the Rao-Blackwell estimator
must be parameter-independent for the measurement to be realizable, so
there is a stringent requirement on the $\mc G$ channel for the
improvement to be realizable, let alone significant.

In practice, the classical Rao-Blackwellization is typically used to
improve an initial estimator design that is not expected to be optimal
or even good in any sense; the derivation of the U-statistics
\cite{vaart} is a representative example. If the initial estimator is
already optimal in the Bayesian sense, then the Rao-Blackwellization
cannot offer any improvement almost everywhere with respect to the
prior $P_X$.

\begin{corollary}
\label{cor_bayes_rb}
Assume the Bayesian problem specified by Eqs.~(\ref{bayes1}) and
(\ref{bayes2}) and let $B = \gce{\mc F}^{\sigma} A$ be the Personick
estimator.  If another CPTP map $\mc G$ is applied and both $\mc G$
and the Rao-Blackwell estimator $\gce{\mc G}^{\rho_x}B$ in
Theorem~\ref{thm_qrbt} do not depend on $x$, then
$\gce{\mc G}^{\rho_x}B$ is a solution to the final Personick
estimator $\gce{\mc G}^{\mc F\sigma} B$.  Moreover, both the
Bayesian error and the local error remain the same after the $\mc G$
channel, in the sense of
\begin{align}
  D_{\sigma,\mc F}(A,\gce{\mc F}^{\sigma} A)
  &= \sum_x P_X(x) \MSE_x = \sum_x P_X(x) \MSE_x' 
    \nonumber\\
  &= D_{\sigma,\mc G\mc F}(A,\gce{\mc G}^{\mc F\sigma} \gce{\mc F}^{\sigma} A),
\label{bayes_suff}
\\
\MSE_x &= \MSE_x' \quad \textrm{almost everywhere $P_X$}.
\label{mse_support}
\end{align}
\end{corollary}
See Appendix~\ref{app_proofs} for the proof.

\section{\label{app_sinha}Comparison with some prior works}
Although Sinha's formalism in Ref.~\cite{sinha22} is applicable to
infinite dimensions and any convex loss function, he makes heavy
assumptions about the commutativity of all the involved
operators. To explain, consider his definition of a
  sufficient-statistic observable $S \in \mc O(\mc H_2)$, which is key
  to his quantum Rao-Blackwell theorems.  Assuming that the spectrum
  of $S$, denoted by $\Lambda_S \subseteq \mathbb R$, is countable for
  simplicity, $S$ can be expressed in the spectral representation
\begin{align}
S &= \sum_{s \in \Lambda_S} s \Pi_s,
\label{S}
\end{align}
where $\Pi$ is a projection-valued measure. Consider any observable
$B$ that commutes with $S$. By the spectral theorem (see Chap.~VII in
Ref.~\cite{reed_simon}), $B$ commutes with all $\{\Pi_s\}$, meaning
that $B$ can be expressed as
\begin{align}
B &= \sum_s \Pi_s B \Pi_s = \bigoplus_s B_s,
\label{Bt}
\end{align}
where $B_s \in \mc O(\Pi_s \mc H_2)$ is equal to $\Pi_s B \Pi_s$ on
the subspace $\Pi_s \mc H_2$. Given a family of density operators
$\{\rho_x:x \in \mc X\}$, Sinha defines a sufficient $S$ by the
existence of a positive function
$\varphi:\Lambda_S\times \mc X \to \mathbb R_+$ and a
parameter-independent positive-semidefinite operator
$\sigma_s \in \mc O(\Pi_s \mc H_2)$ such that, for any $B$ that
commutes with $S$,
\begin{align}
\trace \rho_x B &= \sum_s \varphi(s,x) 
\trace \sigma_s B_s.
\label{sinha}
\end{align}
Since Sinha considers only POVMs
that commute with $S$ in his Rao-Blackwell theorems (see Theorems~4.4
and 5.2 in Ref.~\cite{sinha22}), Eq.~(\ref{sinha}) implies that each
$\rho_x$ can be factorized as
\begin{align}
\rho_x &= \bigoplus_s \varphi(s,x) \sigma_s,
\label{sinha_family}
\end{align}
meaning that $\{\rho_x\}$ all commute with $S$ (see Remark~3.3 in
Ref.~\cite{sinha22}). In fact, Eq.~(\ref{sinha_family}) also implies
that $\{\rho_x\}$ commute with one another. Such assumptions are
extremely restrictive, as noncommutativity is precisely what
distinguishes quantum probability theory from the classical version,
and Sinha's restrictions to it are simply unprecedented in quantum
metrology \cite{helstrom,holevo11,hayashi}. This work, on the other
hand, does not impose any commutativity requirements on the
operators. The key advance here is the use of the GCE formalism in
Secs.~\ref{sec_review} and \ref{sec_prop} that generalizes classical
probability theory from the ground level for noncommuting operators,
so that Sinha's commutativity assumptions are never necessary.

Sinha also avoids any explicit use of quantum conditional expectations
(see Remark~5.3(i) in Ref.~\cite{sinha22}) or even CPTP maps. The use
of a GCE in this work, on the other hand, makes Theorem~\ref{thm_qrbt}
a more natural generalization of the classical theorem. As the
conditional expectation is a standard and crucial step in the
classical Rao-Blackwellization, the GCE can be similarly instrumental
for the quantum case, as demonstrated by the corollaries and examples
in this paper.

It is straightforward to show that a measurement of Sinha's
sufficient-statistic observable is a special case of a sufficient
channel.
\begin{proposition}
\label{prop_sinha}
Given Eqs.~(\ref{Bt}) and (\ref{sinha_family})
and assuming the measurement map
\begin{align}
\mc G \rho_x &= \sum_s (\trace \Pi_s \rho_x \Pi_s) \Pi_s,
\label{map_sinha}
\end{align}
which represents the state of the classical random variable obtained
by a measurement of $S$, a solution to any GCE is
\begin{align}
\gce{\mc G} B &= \sum_s\frac{\trace \sigma_s B_s}{\trace \sigma_s} \Pi_s,
\label{GCE_sinha}
\end{align}
which is independent of $x$.
\end{proposition}
See Appendix~\ref{app_proofs} for the proof.

Equation~(\ref{GCE_sinha}) is used implicitly in Sinha's Rao-Blackwell
theorems and Theorem~\ref{thm_qrbt} here also applies to it.
Physically, Proposition~\ref{prop_sinha} means that, when the density
operators are given by Eq.~(\ref{sinha_family}) and the original
estimator is given by Eq.~(\ref{Bt}), one can measure the
sufficient-statistic observable $S$ given by Eq.~(\ref{S}), and the
Rao-Blackwell estimator as a function of the outcome $s$ is
$(\trace \sigma_s B_s)/\trace\sigma_s$.

Another relevant prior work is Ref.~\cite{luczak15} by \L{}uczak,
which studies a concept of sufficiency in von Neumann algebra for
minimum-variance unbiased estimation in Sec.~5 of
Ref.~\cite{luczak15}. His Theorem~5.1 states that a subalgebra with a
special property called completeness is sufficient for the estimation
if and only if there exists a constant GCE in terms of the Jordan
product that projects onto the subalgebra. He makes no commutativity
assumptions like Sinha's, but the completeness assumption is
unfortunately rather restrictive, as is well known in classical
statistics \cite{lehmann98} and recognized by \L{}uczak himself
\cite{luczak15}. Even in classical statistics, completeness is
difficult to check, and not many models are known to satisfy it. It is
unclear what quantum models beyond the known classical cases can
satisfy the property.  Theorem~\ref{thm_qrbt} here, on the other hand,
does not require the unbiasedness and completeness assumptions.

Lastly, it is worth mentioning that Refs.~\cite{shmaya05,chefles09}
by Shmaya and Chefles concern a quantum generalization of another
Blackwell theorem, which, to my knowledge, has no relation to the
Rao-Blackwell theorem, apart from Blackwell's name being attached to
both.

\section{\label{app_U}Quantum U-statistics}
The goal here is to compute the GCE given by Eq.~(\ref{haar_perm}) for
an operator in the form of Eq.~(\ref{B_perm}).  A few definitions are
necessary before I can proceed. Let
\begin{align}
\BK{e(u): u \in \mc U}
\end{align}
be an orthonormal basis of $\mc O(\mc H_1)$ and
\begin{align}
\BK{E(\bs u) \equiv e(u_1)\otimes \dots \otimes e(u_n): \bs u \in \mc U^n}
\label{E}
\end{align}
be an orthonormal basis of $\mc O(\mc H_1^{\otimes n})$, where $\bs u$
is a column vector and the orthonormality relations are
\begin{align}
\Avg{e(u),e(v)}_{\HS} &= \delta_{uv},
&
\Avg{E(\bs u),E(\bs v)}_{\HS} &= \delta_{\bs u\bs v}.
\end{align}
For example, one can assume the matrix units
$e(u) = \ket{u'}\bra{u''}$ with $u = (u',u'')$.  Any
$B \in \mc O(\mc H_1^{\otimes n})$ can be expressed as
\begin{align}
B &= \sum_{\bs u} B(\bs u) E(\bs u),
&
B(\bs u) &= \Avg{E(\bs u),B}_{\HS}.
\end{align}
Define the permutation matrix $\hat\pi$ on a column vector as
\begin{align}
\hat\pi_{jk} &\equiv \delta_{j \pi(k)},
&
(\hat\pi\bs u)_j &= u_{\pi^{-1}j}.
\end{align}
Then
\begin{align}
U_\pi E(\bs u)U_\pi^\dagger &= E(\hat\pi\bs u),
\end{align}
and the symmetrization map given by Eq.~(\ref{gce_ultimate}) becomes
\begin{align}
\frac{1}{n!} \sum_\pi U_\pi B U_\pi^\dagger
&= \frac{1}{n!} \sum_\pi \sum_{\bs u} B(\bs u) E(\hat\pi \bs u)
\\
&= \sum_{\bs u} \tilde B(\bs u) E(\bs u),
\\
\tilde B(\bs u) &= \frac{1}{n!} \sum_\pi B(\hat\pi\bs u),
\label{gce_perm_mat}
\end{align}
which boils down to a symmetrization of $B(\bs u)$.  In general, a
symmetric operator on $\mc H_1^{\otimes m}$ is defined by
\begin{align}
U_{\pi_m} C U_{\pi_m}^\dagger &= C,
&
C(\hat\pi_m\bs v) &= C(\bs v)
\quad
\forall \pi_m \in S_m.
\end{align}
Given any operator on $\mc H_1^{\otimes m}$, a symmetric version can
be obtained by applying the symmetrization map.

Define a projection matrix $\Pi_{\bs j}:\mc U^n \to \mc U^{\dim \bs j}$
by
\begin{align}
\Pi_{\bs j}\bs u =  \mqty(u_{j_1}\\ \vdots\\ u_{j_m}),
\label{project_j}
\end{align}
where $\bs j = (j_1,\dots,j_m) \in \mc J_m$ is a vector of indices
with $1 \le m \le n$ and $\mc J_m$ is the set of $m$-permutations of
$\{1,\dots,n\}$ (ordered sampling without replacement). Define also
$\{\bs j\}$ for a $\bs j \in \mc J_m$ as the vector of indices sorted
in ascending order and define the set of all such vectors as
\begin{align}
\mc K_m &\equiv \BK{\bs k \in \mc J_m: \bs k = \{\bs k\}},
\label{combs}
\end{align}
which is equivalent to the set of $m$-combinations of $\{1,\dots,n\}$
(unordered sampling without replacement).

A formula for the symmetrization can now be presented.
\begin{proposition}
\label{prop_U}
Suppose that $B \in \mc O(\mc H_1^{\otimes n})$ can be decomposed as
\begin{align}
B &= C \otimes C',
\label{B_decomp}
\end{align}
where $C \in \mc O(\mc H_1^{\otimes m})$ applies to the first $m$
Hilbert subspaces in $\mc H_1^{\otimes n}$ and
$C' \in \mc O(\mc H_1^{\otimes (n-m)})$ applies to the rest.  Assume
that both $C$ and $C'$ are symmetric.  Then the symmetrized $B$ is
given by
\begin{align}
\frac{1}{n!} \sum_\pi U_\pi B U_\pi^\dagger
&= \mqty(n\\ m)^{-1}\sum_{\bs k \in \mc K_{m}} (C\otimes C')_{\bs k},
\nonumber\\
(C\otimes C')_{\bs k} &\equiv 
\sum_{\bs u} C(\Pi_{\bs k}\bs u)C'(\Pi_{\bs k'}\bs u) E(\bs u),
\label{U}
\end{align}
where $\{E(\bs u)\}$ is an orthonormal basis of
$\mc O(\mc H_1^{\otimes n})$ given by Eq.~(\ref{E}), $C(\bs v)$ and
$C'(\bs w)$ are the components of $C$ and $C'$ with respect to the
same basis, the projection matrix $\Pi$ is defined by
Eq.~(\ref{project_j}), $\mc K_m$ is the $m$-combinations of
$\{1,\dots,n\}$, and for each $\bs k$, $\bs k'$ is defined as the rest
of the indices in $\{1,\dots,n\}$.
\end{proposition}
See Appendix~\ref{app_proofs} for the proof.

Each $(C\otimes C')_{\bs k}$ in Eqs.~(\ref{U}) is an application of
$C$ on the $m$ Hilbert subspaces in $\mc H_1^{\otimes n}$ indexed by
$\bs k = (k_1,\dots,k_m)$ and an application of $C'$ on the other
$n-m$ Hilbert subspaces.  If $C$ is not symmetric, it can be
symmetrized first before Proposition~\ref{prop_U} is used. This is
because the left Haar measure is also the right Haar measure for the
permutation group, making the symmetrization map invariant to any
prior permutation as well. One is therefore free to symmetrize $C$ in
$C\otimes C'$ first before the total symmetrization in
Proposition~\ref{prop_U}. The same goes for $C'$.  If $B$ is in the
general form of $\otimes_n C_n$, Proposition~\ref{prop_U} can be
applied recursively to produce a generalized multinomial form of
Eqs.~(\ref{U}).

Proposition~\ref{prop_U} gives the quantum U-statistics in
Ref.~\cite{guta10} if Eq.~(\ref{B_perm}), a special case of
Eq.~(\ref{B_decomp}) with $C' = I_1^{\otimes(n-m)}$, is assumed.  The
classical U-statistics \cite{lehmann98,vaart} are obtained by assuming
$\{e(u) = \ket{u}\bra{u}\}$, where $\{\ket{u}\}$ is an orthonormal
basis of $\mc H_1$, and $C'(\Pi_{\bs k'}\bs u) = 1$, such that the
estimator in terms of the classical variable $\bs u$ becomes
\begin{align}
\tilde B(\bs u) &= \mqty(n\\ m)^{-1}
\sum_{\bs k \in \mc K_{m}} C(\Pi_{\bs k}\bs u).
\end{align}

\section{\label{app_proofs}Proofs}
\begin{proof}[Proof of Lemma~\ref{lem_tot}]
  To prove the first and last equalities in Eq.~(\ref{mean}), write
\begin{align}
\Avg{I_1,A}_\sigma &= \Avg{I_1,\mc E_\sigma A}_{\HS}
= \Avg{\mc E_\sigma I_1,A}_{\HS}
\nonumber\\
&= \Avg{\sigma,A}_{\HS} = \trace \sigma A,
\end{align}
where the self-adjoint property of $\mc E_\sigma$ and Eq.~(\ref{ax1})
have been used. To prove the second equality in Eq.~(\ref{mean}), plug
$c = I_2$ into Eq.~(\ref{gce}) and use the unital property
$\mc F^* I_2 = I_1$.
\end{proof}

\begin{proof}[Proof of Lemma~\ref{lem_var}]
Write
\begin{align}
\norm{A - a I_1}_\sigma^2 &= \norm{A}_\sigma^2 - 2 \Re\Bk{a^*\Avg{I_1,A}_\sigma}
+ |a|^2,
\label{err1}
\\
\norm{\gce{\mc F} A - a I_2}_{\mc F\sigma}^2 &= \norm{\gce{\mc F} A}_{\mc F\sigma}^2 
- 2 \Re\Bk{a^*\Avg{I_2,\gce{\mc F} A}_{\mc F\sigma}}
\nonumber\\&\quad + |a|^2.
\label{err2}
\end{align}
Lemma~\ref{lem_tot} gives
$\avg{I_1,A}_\sigma = \avg{I_2,\gce{\mc F} A}_{\mc F\sigma}$.  Then
Eq.~(\ref{err1}) minus Eq.~(\ref{err2}) gives Eq.~(\ref{RB}) via
Eq.~(\ref{Dmin}).
\end{proof}

\begin{proof}[Proof of Lemma~\ref{lem_sep}]
  Equation~(\ref{gce}) gives, for any $c \in L_2(\sigma_x)$,
\begin{align}
\Avg{c,\gce{\mc G}B}_{\sigma_x} &= \Avg{\mc G^* c,B}_{\rho_x}
\\
&= \trace \Bk{c^\dagger \trace_0 \bk{\mc E_{\sigma_x\otimes\tau} B}}
\\
&= \trace \Bk{(c\otimes I_0)^\dagger \mc E_{\sigma_x\otimes\tau} B}
\\
&= \Avg{c\otimes I_0,\mc E_{\sigma_x\otimes\tau}  B  }_{\HS}
\\
&=  \Avg{\mc E_{\sigma_x\otimes\tau}(c\otimes I_0),  B }_{\HS}
\\
&= \Avg{(\mc E_{\sigma_x}c)\otimes (\mc E_\tau I_0),  B }_{\HS}
\\
&= \Avg{(\mc E_{\sigma_x}c)\otimes \tau,  B }_{\HS}
\\
&= \trace\BK{[(\mc E_{\sigma_x} c)^\dagger\otimes \tau]  B }
\\
&= \trace \BK{(\mc E_{\sigma_x} c)^\dagger \trace_0[(I_1\otimes \tau)  B ]}
\\
&= \Avg{\mc E_{\sigma_x} c,\trace_0[(I_1\otimes \tau) B ]}_{\HS}
\\
&= \Avg{c,\trace_0[(I_1\otimes \tau)B ]}_{\sigma_x},
\label{lem_sep_proof}
\end{align}
where the self-adjoint property of $\mc E$ and Eqs.~(\ref{ax1}) and
(\ref{ax3}) have been used at various steps.
Equation~(\ref{lem_sep_proof}) means that
$\trace_0[(I_1\otimes \tau) B]$ is a solution to the GCE
$\gce{\mc G}B$.
\end{proof}

\begin{proof}[Proof of Corollary~\ref{cor_vn}]
  Corollary~\ref{cor_vn2} states that, given the local error $\MSE_x$
  for any POVM $M$ and any estimator $b$, there exists an
  operator-valued estimator on $\mc H_2$ with an error $\MSE_x'$ that
  satisfies $\MSE_x \ge \MSE_x'$ for all $x \in \mc X$.  The average
  error of $(M,b)$ is then also bounded as
\begin{align}
\sum_x P_X(x) \MSE_x \ge \sum_x P_X(x) \MSE_x' \ge
D_{\sigma,\mc F}(A,\gce{\mc F} A),
\end{align}
where the last inequality follows from the optimality of the Personick
estimator $\gce{\mc F} A$ among all observables on $\mc H_2$, as per
Definition~\ref{def_gce}.

\end{proof}

\begin{proof}[Proof of Lemma~\ref{lem_invar}]
\begin{align}
\mc G \mc E_{\rho_x} B &= 
\int \mu(dz) U_z \bk{\mc E_{\rho_x} B} U_z^\dagger
\\
&= \int \mu(dz)
 \mc E_{U_z \rho_x U_z^\dagger}(U_z B U_z^\dagger)
\\
&= \mc E_{\rho_x}\int \mu(dz) U_z B U_z^\dagger,
\label{lem_perm_proof}
\end{align}
where Eqs.~(\ref{ax2}) and (\ref{rho_invar}) have been used.  The
interchange of $\mc E_{\rho_x}$ and the Bochner integral is valid
because $\mc E_{\rho_x}$ is a linear map on a finite-dimensional
operator space (more assumptions would be needed for
infinite-dimensional operator spaces; see Corollary~2 on p.~134 in
Ref.~\cite{yosida}). By Eq.~(\ref{gce_exp}),
Eq.~(\ref{lem_perm_proof}) is equal to
$\mc E_{\mc G\rho_x} \gce{\mc G} B = \mc E_{\rho_x} \gce{\mc G}
B$, resulting in a solution to the GCE given by Eq.~(\ref{gce_invar}).
\end{proof}

\begin{proof}[Proof of Lemma~\ref{lem_sum}]
Assuming Eq.~(\ref{gce_sum}) and using Eq.~(\ref{E_sum}), one obtains
\begin{align}
\mc E_{\rho_x}  \gce{\mc G} B &= \bigoplus_n \mc E_{\sigma_x^{(n)}}\bk{\Pi_n B \Pi_n},
\end{align}
which is equal to
\begin{align}
\mc G \mc E_{\rho_x} B &= \bigoplus_n \Pi_n \bk{\mc E_{\rho_x}B}\Pi_n
= \bigoplus_n \mc E_{\sigma_x^{(n)}} \bk{\Pi_n B \Pi_n},
\end{align}
by virtue of Eq.~(\ref{E_project}). It follows that
Eq.~(\ref{gce_sum}) is a solution to the GCE, as per
Eq.~(\ref{gce_exp}).
\end{proof}
Note that Lemmas~\ref{lem_tot}--\ref{lem_sum} apply to classes of GCEs
and not just the Jordan version.  Note also that the GCEs for any
sequence of the channels can be computed by chaining the individual
GCEs in a manner reminiscent of calculus.

\begin{proof}[Proof of Proposition~\ref{prop_gauss}]
  The GCE defined by Eq.~(\ref{gce_exp}) can be solved by the operator
  Fourier transform
\begin{align}
\trace \bk{\mc E_{\mc F\sigma} \gce{\mc F} A} \tilde W(\zeta)
&= \trace \bk{\mc F \mc E_\sigma A} \tilde W(\zeta).
\label{gce_char}
\end{align}
The right-hand side can be expressed as
\begin{align}
\trace \bk{\mc F \mc E_\sigma A} \tilde W(\zeta)
&= \trace (\mc E_\sigma A) \mc F^* \tilde W(\zeta)
\\
&= f(\zeta) \trace (\mc E_\sigma A) W(F^\top \zeta)
\label{step_F}
\\
&= f(\zeta) \Bk{-i u^\top \nabla \phi(z)}_{z = F^\top\zeta}
\label{weyl_diff}
\\
&= u^\top\bk{ m +i \Sigma F^\top\zeta}
\tilde\phi(\zeta),
\label{step_char}
\end{align}
where Eq.~(\ref{step_F}) has used the Gaussian-channel definition
given by Eq.~(\ref{F_gauss}), Eq.~(\ref{weyl_diff}) has used
Eq.~(5.4.43) in Ref.~\cite{holevo11} with
$\nabla \equiv \mqty(\pdv*{}{x_1}& \pdv*{}{y_1} & \dots & \pdv*{x_s} &
\pdv*{y_s})^\top$, and Eq.~(\ref{step_char}) has used the Gaussian
$\phi(z)$ given by Eq.~(\ref{phi_gauss}) and the output
$\tilde\phi(\zeta)$ given by Eq.~(\ref{phi_tilde}).  With similar
steps and the ansatz
\begin{align}
\gce{\mc F} A &= v^\top \tilde Q + c,
&
v &\in \mathbb R^{2t},
&
c &\in \mathbb R,
\label{gce_ansatz}
\end{align}
the left-hand side of Eq.~(\ref{gce_char}) can be expressed as
\begin{align}
\trace \bk{\mc E_{\mc F\sigma} \gce{\mc F} A} \tilde W(\zeta)
&= \Bk{v^\top (\tilde m + i \tilde\Sigma \zeta) + c}\tilde\phi(\zeta),
\label{lhs_char}
\end{align}
where Eqs.~(\ref{phi_tilde}) and (\ref{phi_tilde2}) are assumed.
Equating Eq.~(\ref{step_char}) with Eq.~(\ref{lhs_char}) leads to
\begin{align}
v^\top &= u^\top\Sigma F^\top \tilde\Sigma^{-1},
\label{v}
\\
c &= u^\top m  - v^\top\tilde m.
\label{c}
\end{align}
Equations~(\ref{v}) and (\ref{c}) can then be substituted into
Eq.~(\ref{gce_ansatz}) to give Eqs.~(\ref{gce_gauss}) and
(\ref{kalman_gain}) via Eqs.~(\ref{m_tilde}) and (\ref{Sigma_tilde}).

To derive Eq.~(\ref{D_gauss}), use Lemmas~\ref{lem_tot} and \ref{lem_var}
to write
\begin{align}
a &= \trace \sigma A = \trace (\mc F\sigma) (\gce{\mc F} A),
\\
D_{\sigma,\mc F}(A,\gce{\mc F} A)
&= \norm{ A - a I_1}_\sigma^2 - \norm{\gce{\mc F} A - a I_2}_{\mc F\sigma}^2
\\
&= u^\top \Sigma u - v^\top \tilde\Sigma v,
\label{D_step}
\end{align}
where the last step has used the fact that $A$ and $\gce{\mc F} A$
are both quadrature operators and their variances are determined by
the covariance matrices of the Gaussian states. Substituting
Eqs.~(\ref{v}) and (\ref{Sigma_tilde}) into Eq.~(\ref{D_step}) leads to
Eq.~(\ref{D_gauss}).
\end{proof}

\begin{proof}[Proof of Corollary~\ref{cor_bayes_rb}]
  Let $c$ be any operator on $\mc H_3$. By the definition of
  $\gce{\mc G}^{\rho_x} B$ given by Eq.~(\ref{gce}),
\begin{align}
\Avg{c,\gce{\mc G}^{\rho_x} B}_{\mc G\rho_x} &= \Avg{\mc G^* c,B}_{\rho_x}
\quad
\forall x \in \mc X.
\end{align}
Taking the expectation of this equation with respect to $P_X(x)$, 
one obtains
\begin{align}
\sum_x P_X(x) \Avg{c,\gce{\mc G}^{\rho_x} B}_{\mc G\rho_x} &= 
\sum_x P_X(x) \Avg{\mc G^* c,B}_{\rho_x},
\\
\Avg{c,\gce{\mc G}^{\rho_x} B}_{\mc G\mc F\sigma} &= 
\Avg{\mc G^* c,B}_{\mc F\sigma},
\label{gce_rb_bayes}
\end{align}
where Eq.~(\ref{gce_rb_bayes}) has used the facts that $c$, $B$,
$\mc G$, and $\gce{\mc G}^{\rho_x}B$ all do not depend on $x$, the
trace and $\mc G$ are linear, the Jordan product is bilinear, and
$\sum_x P_X(x)\rho_x = \mc F\sigma$. Equation~(\ref{gce_rb_bayes}) means
that $\gce{\mc G}^{\rho_x}B$ satisfies the definition of the final
Personick estimator $\gce{\mc G}^{\mc F\sigma} B$ as per
Eq.~(\ref{gce}).

Equation~(\ref{bayes_suff}) can be proved by combining the
monotonicity of the Bayesian error (Corollary~\ref{cor_mono}) and the
quantum Rao-Blackwell theorem (Theorem~\ref{thm_qrbt}).

Equation~(\ref{mse_support}) can be proved by contradiction: assume
that there exists a $x \in \mc X$ with $P_X(x) > 0$ such that
$\MSE_x > \MSE_x'$.  Since $\MSE_x \ge \MSE_x'$ by
Theorem~\ref{thm_qrbt}, the assumption would imply
$\sum_x P_X(x) \MSE_x > \sum_x P_X(x) \MSE_x'$, which contradicts
Eq.~(\ref{bayes_suff}).  It follows that the assumption cannot hold
and one must have Eq.~(\ref{mse_support}).
\end{proof}

\begin{proof}[Proof of Proposition~\ref{prop_sinha}]
Given Eq.~(\ref{map_sinha}), a solution to any GCE is
\begin{align}
\gce{\mc G} B &= 
\sum_s \frac{\trace \Pi_s \mc E_{\rho_x} B \Pi_s}{\trace \Pi_s \rho_x \Pi_s}\Pi_s.
\label{GCE_sinha2}
\end{align}
The numerator can be expressed as
\begin{align}
\trace \Pi_s \mc E_{\rho_x} B \Pi_s
&= \avg{\Pi_s,\mc E_{\rho_x} B}_{\HS}
= \avg{\mc E_{\rho_x}\Pi_s, B}_{\HS}
\\
&= \avg{\rho_x \Pi_s,B}_{\HS}
= \varphi(s,x) \trace \sigma_s B_s,
\end{align}
where the self-adjoint property of $\mc E_{\rho_x}$, the commutativity
between $\rho_x$ and $\Pi_s$, and Eqs.~(\ref{ax1}), (\ref{Bt}), and
(\ref{sinha_family}) have been used at various steps. Similarly, the
denominator in Eq.~(\ref{GCE_sinha2}) can be expressed as
\begin{align}
\trace \Pi_s \rho_x \Pi_s &= \varphi(s,x) \trace \sigma_s.
\end{align}
Equation~(\ref{GCE_sinha}) then follows.
\end{proof}

\begin{proof}[Proof of Proposition~\ref{prop_U}]
With Eq.~(\ref{B_decomp}), $B(\bs u)$ becomes
\begin{align}
B(\bs u) &= C(\Pi_{[1,m]}\bs u) C'(\Pi_{[m+1,n]}\bs u),
\\
[1,m] &\equiv (1,\dots,m),
\\
[m+1,n] &\equiv (m+1,\dots,n).
\end{align}
With the identity
\begin{align}
\Pi_{\bs j}\hat\pi\bs u &= \Pi_{\pi^{-1}\bs j}\bs u,
\end{align}
$B(\hat\pi\bs u)$ in Eq.~(\ref{gce_perm_mat}) becomes
\begin{align}
B(\hat\pi\bs u) &= C(\Pi_{[1,m]}\hat\pi\bs u) 
C'(\Pi_{[m+1,n]}\hat\pi\bs u)
\\
&= C(\Pi_{\pi^{-1}[1,m]}\bs u) C'(\Pi_{\pi^{-1}[m+1,n]}\bs u).
\end{align}
The symmetry of $C$ and $C'$ implies
\begin{align}
C(\Pi_{\bs j}\bs u)  &= C(\Pi_{\{\bs j\}}\bs u),
\quad
\forall \bs j \in \mc J_m,
\\
C'(\Pi_{\bs j}\bs u) &= C'(\Pi_{\{\bs j\}}\bs u),
\quad
\forall \bs j \in \mc J_{n-m},
\\
B(\hat\pi\bs u) &= C(\Pi_{\{\pi^{-1}[1,m]\}}\bs u) C'(\Pi_{\{\pi^{-1}[m+1,n]\}}\bs u).
\end{align}
The $n!$ summands in Eq.~(\ref{gce_perm_mat}) with respect to $\pi$
can now be divided into subsets indexed by Eq.~(\ref{combs}). Each
subset, indexed by a $\bs k \in \mc K_m$, contains $m!(n-m)!$ terms
all equal to $C(\Pi_{\bs k}\bs u)C'(\Pi_{\bs k'}\bs u)$ with
\begin{align}
\bs k &= \BK{\pi^{-1}[1,m]},
&
\bs k' &= \BK{\pi^{-1}[m+1,n]}.
\end{align}
The sum in Eq.~(\ref{gce_perm_mat}) becomes
\begin{align}
\frac{1}{n!}\sum_\pi B(\hat\pi\bs u)
&= 
\mqty(n\\ m)^{-1}\sum_{\bs k \in \mc K_{m}} C(\Pi_{\bs k}\bs u)C'(\Pi_{\bs k'}\bs u),
\label{perm_decomp}
\end{align}
where $\mqty(n\\ m) \equiv n!/m!(n-m)! =|\mc K_m|$ is the binomial
coefficient.  The proposition hence follows.
\end{proof}

\bibliographystyle{unsrtnat}
\bibliography{research2}

%\pagebreak

%\include{resubmission_letter_20230824}

\end{document}